\documentclass[11pt]{amsart}
\usepackage[margin=2.5cm]{geometry}

\usepackage{enumerate}
\usepackage{amsmath}
\usepackage{amssymb,latexsym}
\usepackage{amsthm}
\usepackage{color}
\usepackage[dvipsnames]{xcolor}
\usepackage{cancel}
\usepackage{graphicx}
\usepackage{cite}
\usepackage{changes}
\makeatletter
\renewcommand{\tocsection}[3]{%
  \indentlabel{\@ifnotempty{#2}{\bfseries\ignorespaces#1 #2\quad}}\bfseries#3}
\renewcommand{\tocsubsection}[3]{%
  \indentlabel{\@ifnotempty{#2}{\ignorespaces#1 #2\quad}}#3}

\newcommand\@dotsep{4.5}
\def\@tocline#1#2#3#4#5#6#7{\relax
  \ifnum #1>\c@tocdepth 
  \else
    \par \addpenalty\@secpenalty\addvspace{#2}%
    \begingroup \hyphenpenalty\@M
    \@ifempty{#4}{%
      \@tempdima\csname r@tocindent\number#1\endcsname\relax
    }{%
      \@tempdima#4\relax
    }%
    \parindent\z@ \leftskip#3\relax \advance\leftskip\@tempdima\relax
    \rightskip\@pnumwidth plus1em \parfillskip-\@pnumwidth
    #5\leavevmode\hskip-\@tempdima{#6}\nobreak
    \leaders\hbox{$\m@th\mkern \@dotsep mu\hbox{.}\mkern \@dotsep mu$}\hfill
    \nobreak
    \hbox to\@pnumwidth{\@tocpagenum{\ifnum#1=1\bfseries\fi#7}}\par
    \nobreak
    \endgroup
  \fi}
\AtBeginDocument{%
\expandafter\renewcommand\csname r@tocindent0\endcsname{0pt}
}
\def\l@subsection{\@tocline{2}{0pt}{2.5pc}{5pc}{}}
\makeatother

\usepackage{etoolbox}
\makeatletter
\patchcmd{\@setaddresses}{\indent}{\noindent}{}{}
\patchcmd{\@setaddresses}{\indent}{\noindent}{}{}
\patchcmd{\@setaddresses}{\indent}{\noindent}{}{}
\patchcmd{\@setaddresses}{\indent}{\noindent}{}{}
\makeatother

\usepackage[linktocpage]{hyperref}
\hypersetup{
  colorlinks   = true, 
  urlcolor     = blue, 
  linkcolor    = Purple, 
  citecolor   = red 
}
\usepackage{comment}
\usepackage{amscd}
\usepackage{mathtools}

\DeclareMathOperator{\C}{\mathcal{C}}

\DeclareMathOperator{\Gal}{Gal}

\DeclareMathOperator{\rk}{rk}

\DeclareMathOperator{\dd}{d}
\DeclareMathOperator{\Mat}{Mat}

\DeclareMathOperator{\GL}{GL}

\DeclareMathOperator{\ww}{w}

\theoremstyle{definition}
\newtheorem{theorem}{Theorem}[section]

\newtheorem{lemma}[theorem]{Lemma}
\newtheorem{corollary}[theorem]{Corollary}
\newtheorem{definition}[theorem]{Definition}
\newtheorem{proposition}[theorem]{Proposition}

\newtheorem{remark}[theorem]{Remark}

\newtheorem{open}[theorem]{Problem}

\newcommand{\fqn}{\mathbb{F}_{q^n}}

\newcommand{\cC}{{\mathcal C}}

\newcommand{\F}{{\mathbb F}}

\newcommand{\Z}{{\mathbb Z}}
\newcommand{\NN}{{\mathbb N}}

\newcommand{\U}{{\mathcal{U}}}
\newcommand{\muu}{{\boldsymbol{\mu}}}

\newcommand{\ds}{\mathrm{d}^{\mathrm{srk}}}

\newcommand{\bfn}{\mathbf {n}}

\newcommand{\fq}{{\mathbb F}_{q}}
\newcommand{\Fq}{{\mathbb F}_{q}}
\newcommand{\Fm}{{\mathbb F}_{q^m}}

\newcommand{\la}{\langle}
\newcommand{\ra}{\rangle}
\renewcommand{\mod}{\hbox{{\rm mod}\,}}

\newcommand{\PG}{\mathrm{PG}}
\newcommand{\N}{\mathrm{N}}

\newcommand\qbin[3]{\left[\begin{matrix} #1 \\ #2 \end{matrix} \right]_{#3}}

\newcommand{\Fmnk}{[\bfn,k]_{q^m/q}}
\newcommand{\Fmnkd}{[\bfn,k,d]_{q^m/q}}

\newcommand{\Fmkd}{[n,k,d]_{q^m/q}}

\newcommand{\st}{\,:\,}

\title{On MSRD codes, \lowercase{h}-designs and disjoint maximum scattered linear sets}
\date{}
\author[P. Santonastaso]{Paolo Santonastaso}
\address{Paolo Santonastaso, \textnormal{Dipartimento di Matematica e Fisica, Universit\`a degli Studi della Campania ``Luigi Vanvitelli'', Viale Lincoln, 5, I--\,81100 Caserta, Italy}}
\email{paolo.santonastaso@unicampania.it}

\author[J. Sheekey]{John Sheekey}
\address{John Sheekey, \textnormal{School of Mathematics and Statistics, UCD,  Belfield, Dublin 4}}
\email{john.sheekey@ucd.ie}

\begin{document}

\maketitle

\begin{abstract}
In this paper we study geometric aspects of codes in the sum-rank metric. We establish the geometric description of generalised weights, and analyse the Delsarte and geometric dual operations. We establish a correspondence between maximum sum-rank distance codes and $h$-designs, extending the well-known correspondence between MDS codes and arcs in projective spaces and between MRD codes and $h$-scatttered subspaces. We use the geometric setting to construct new  $h$-designs and new MSRD codes via new families of pairwise disjoint maximum scattered linear sets. 

\end{abstract}



\section{Introduction}

The interplay between coding theory and finite geometry has a long and rich history. From the classical connections between linear codes in the Hamming metric and (multi)sets of points in a projective space, through more recent connections between linear codes in the rank metric and linear sets in a projective space, these correspondences have lead to many interesting results, techniques, and research questions. Natural notions in one of these settings lead to new questions in the other, and in many cases questions that appear difficult to address in one setting become more tractable in the other.

In this paper we consider codes in the {\it sum-rank metric}, a metric that can be seen as a generalisation of both the Hamming and rank metrics, and their geometric counterparts. The sum-rank metric has attracted sustained interest in recent years, due to applications in information theory, including multishot network coding, distributed storage systems and multi-antenna wireless communication \cite{martinez2022survey}. Codes in the sum-rank metric can be viewed as spaces of block-diagonal matrices, and so this allows for a broad range of cases of possible interest by varying the number and sizes of the blocks; however, most research focuses on cases with additional structure where they more closely resemble the classical cases of the Hamming and rank-metrics, while still generalising enough to give rise to a rich source of interesting mathematical questions and sufficient applicability. 


The correspondence between linear codes in the Hamming metric and sets of points in a projective space (sometimes referred to as {\it projective systems}) is well-known and is a key tool in classical coding theory. More recently, a correspondence between codes in the rank-metric (with additional linearity assumptions) and geometric objects analogous to projective systems has been established, leading to new constructions, techniques, and furthering our understanding of the structure of these codes.  

With certain assumptions on the parameters and linearity of a sum-rank metric code, the geometric counterparts of optimal codes are {\it (optimal) subspace designs} and further analogues of projective systems. These are sets of $\Fq$-subspaces of an $\Fm$-vector space with restricted intersection properties with respect to $\Fm$-subspaces, and were first introduced by Guruswami and Xing in \cite{guruswami2016explicit} for the purposes of studying rank-metric codes. The connection with sum-rank metric codes was established in \cite{neri2023geometry}, which leads to a host of natural questions regarding the existence and structure of these objects.

In this paper we determine the geometric counterpart of the coding-theoretic notion of {\it generalised weights} for the sum-rank metric in terms of subspace designs, and study geometric characterisations of the systems corresponding to {\it MSRD codes}. We study the implications of the coding-theoretic {\it Delsarte dual} on subspace designs, and in the opposite direction we consider the implications of  geometric duality for codes. We then use the geometric setting to construct new subspace designs, and hence new MSRD codes. In particular we show the existence of large sets of of pairwise disjoint {\it maximum scattered linear sets} in many projective spaces.

\subsection{Notation}

Through this paper $q$ is a prime power and $\Fq$ is the finite field with $q$ elements. By $\Fm$ we denote the degree $m$ extension field  of $\Fq$ and by $V(k,q^m)$ we denote a $k$-dimensional $\F_{q^m}$-vector space. For a positive integer $i$, by $[i]$, we denote the set $[i]=\{j \in \mathbb{N}: 1 \leq j \leq i\}$.
\noindent Let $t$ be a positive integer. With
$\bfn=(n_1,\ldots,n_t) \in \NN^t$ we will always denote the ordered tuples with $n_1 \geq n_2 \geq \ldots \geq n_t$ and we set $N\coloneqq n_1+\ldots+n_t$. We use the following compact notation for the direct sum of vector spaces 
$$ \Fm^\bfn\coloneqq\bigoplus_{i=1}^t\Fm^{n_i}.$$
If $A_i \in \F_{q}^{n_i \times n_i}$, for each $i$, we denote by $diag(A_1,\ldots,A_t) \in \F_q^{N \times N}$ the block diagonal matrix having as blocks on the diagonal the $A_i$'s.

\subsection{Layout of paper}

The paper is organised as follows. In Section 2 we recall the basic definitions and preliminary results on sum-rank metric codes and subspace designs. We establish the necessary tools and correspondences between the settings, and prove some basic properties which will be used in later sections. In Section 3 we further develop the geometric setting, in particular with regard to $h$-designs, generalised weights, and geometric duality. Section 4 contains the main new constructions of this paper; in particular we construct new examples of optimal maximum $1$-designs for many new parameters. In Section 5 we show how these objects give rise to two-intersecting sets, strongly regular graphs, and two-weight codes in the Hamming metric. Section 6 contains the coding theory consequences of the constructions of Section 4, as well as analysis of the duality operations (both geometric and coding theoretic) applied to the new constructions.

\section{Sum-rank-metric codes and their geometric description}

\subsection{Preliminaries}

The $\Fq$-rank of a vector $v=(v_1,\ldots,v_n) \in \F_{q^m}^n$ is defined as the quantity
$$\rk_q(v)\coloneqq\dim_{\fq} \langle v_1,\ldots, v_n\rangle_{\fq}.$$ 
For an element $x=(x_1 , \ldots,  x_t)\in\F_{q^m}^\bfn$, with $x_i\in\F_{q^m}^{n_i}$, (sometimes we also write $(x_1 \mid \ldots \mid x_t)$), the \textbf{$\Fq$-sum-rank weight of $x$} is defined as
$$ \ww(x)=\sum_{i=1}^t \rk_q(x_i).$$
A \textbf{(vector) sum-rank metric code} $\C$ is an $\Fm$-subspace of $\Fm^\bfn$ endowed with the 
\textbf{$\Fq$-sum-rank distance on $\Fm^\bfn$} defined as 
\[
\dd(x,y)=\sum_{i=1}^t \rk_q(x_i-y_i),
\]
for any  $x=(x_1 , \ldots, x_t), y=(y_1, \ldots, y_t) \in \F_{q^m}^\bfn$, with $x_i,y_i \in \F_{q^m}^{n_i}$. The minimum (sum-rank) distance of a code $\C$ is denoted $d(\C)$, that is,
\[
d(\C):=\min\{\dd(x,y) \st x, y \in \C, x\neq y  \}= \min\{\ww(x) \st x \in \C, x\neq 0  \}.
\]
From now on, we write that $\C$ is an $\Fmnkd$ code, if $\C$ is a sum-rank metric code in $\F_{q^m}^{\bfn}$ having $\Fm$-dimension $k$ and minimum (sum-rank) distance $d$, or simply an $\Fmnk$ code if the minimum distance is not relevant/known.

\medskip 
For an $\Fmnk$ code, we denote by $W_r(\C)$ the number of codewords having $\fq$-sum-rank weight $r$. For vector sum-rank metric codes, a Singleton-like bound holds.

\begin{theorem}[see \textnormal{\cite[Theorem III.2]{byrne2021fundamental} \cite[Proposition 16]{Martinez2018skew}}]  \label{th:Singletonbound}
    Let $\C$ be an $[\bfn,k,d]_{q^m/q}$ code. Let $j$ and $\delta$ be the unique integers satisfying \[d-1=\sum_{i=1}^{j-1}\min\{m,n_i\}+\delta \mbox{ and }0 \leq \delta \leq \min\{m,n_j\}-1.\]
Then 
\[
\lvert \C \rvert \leq q^{m\sum_{i=j}^tn_i-\max\{m,n_j\} \delta}.\] 
In particular,
    \[
    \begin{cases}
    k \leq N-d+1 & \mbox{ if } m \geq n_1; \\
    mk \leq n(tm-d+1) & \mbox{ if } n\coloneqq n_1=\ldots=n_t \geq m.
    \end{cases}
    \]
\end{theorem}

An $[\bfn,k,d]_{q^m/q}$ code is called \textbf{Maximum Sum-Rank Distance code} (or shortly \textbf{MSRD code}) if its parameters attains the bound of Theorem \ref{th:Singletonbound}.


\medskip

Let consider the symmetric and non-degenerate bilinear form on $\F_{q^m}^\bfn$ 
\[(x,y)\in\F_{q^m}^\bfn\times \F_{q^m}^\bfn \mapsto \sum_{i=1}^t x_i \cdot y_i \in \F_{q^m}\] induced by the standard inner product on $\F_{q^m}^n$. 
The dual of an $[\bfn,k,d]_{q^m/q}$ code $\C$ is defined as 
\[
\C^{\perp} =\left\{
y=(y_1 ,\ldots, y_t) \in \F_{q^m}^\bfn \colon \sum_{i=1}^t x_i \cdot y_i=0 \mbox{ for all } x=(x_1 \mid \ldots \mid x_t) \in \C
\right\}.
\]

Consequentely, the dual of a $\Fmnkd$ code is a $[\bfn,N-k]_{q^m/q}$ code, where $N=\sum_{i=1}^t n_i$. Moreover, the MSRD property is invariant under dualization.

\begin{theorem} [see \textnormal{\cite[Theorem 6.1]{byrne2021fundamental},\cite[Theorem 5]{martinez2019theory}}] \label{th:dualMSRD}
Let $\C$ be an $[\bfn,k,d]_{q^m/q}$ code. 
Then, $$d(\C^{\perp})\leq 
     \begin{cases}
     N-d+2, & \mbox{ if } m \geq n_1, \\
     tm-d+2, & \mbox{ if } n\coloneqq n_1=\ldots=n_t \geq m,
     \end{cases}
     $$
    and the equality holds when $\C$ is an MSRD code. In particular, if $\C$ is an MSRD code then $\C^\perp$ is an MSRD code as well.
\end{theorem}

Now, we recall the definition of equivalence between sum-rank metric codes. We will use the $\Fm$-linear isometries of the whole ambient space in order to define the equivalence of sum-rank metric codes. The $\Fm$-linear isometries of the space $\F_{q^m}^\bfn$ endowed with the sum-rank distance are classified in \cite[Theorem 3.7]{alfarano2021sum}. The case $\bfn=(n,\ldots, n)$ was already proved in \cite[Theorem 2]{martinezpenas2021hamming}.

In this way, we say that two $\Fmnk$ codes $\cC_1, \cC_2$ are \textbf{equivalent} if there is an $\Fm$-linear isometry mapping $\cC_1$ in $ \cC_2$. More precisely, $\cC_1, \cC_2$ are equivalent if and only if there exist an element $(A_1,\ldots,A_t) \in \GL(n_1,\fq) \times \ldots \times \GL(n_t,\fq)$, an element $(a_1,\ldots,a_t) \in (\F_{q^m}^*)^t$ and a permutation $\pi \in \mathcal{S}_t$ such that
\[
\C_2=\left\{(a_1 x_{\pi(1)}A_1 \mid \ldots \mid a_tx_{\pi(t)}A_t) \colon (x_1\mid \ldots \mid x_t) \in \C_1\right\}.
\]

Finally, we recall the notion of degeneracy of a sum-rank metric code.
Let $\C$ be a $\Fmnk$ code. A generator matrix $G$ for $\C$ is naturally partitioned as $G=(G_1 \,|\, \ldots \,|\, G_t)$, with $G_i \in \Fm^{k \times n_i}$. In this way, $\C$ is said to be \textbf{sum-rank nondegenerate} if for any generator matrix $G=(G_1 \,|\, \ldots \,|\, G_t)$, with $G_i \in \Fm^{k \times n_i}$, the columns of each $G_i$ are $\Fq$-linear independent. This property can be characterize as follows.
\begin{proposition} [see \textnormal{\cite[Proposition 2.13]{neri2023geometry}}]\label{prop:charact_nondegenarate}
Let $\C$ be a $\Fmnk$ code with $k<N$. The following are equivalent.
\begin{enumerate}
    \item $\C$ is sum-rank nondegenerate.
    \item $\dd(\C^\perp) \geq 2$.
\end{enumerate}
\end{proposition}

We denote the set of equivalence classes of nondegenerate $\Fmnkd$ codes by $\mathfrak{C}\Fmnkd$.

The interested reader is referred to \cite{Martinez2018skew,martinez2019theory,neri2022twisted,ott2021bounds} for a more detailed description of this setting.

\subsection{Geometry and Delsarte dual of a system}

Recently in \cite{neri2023geometry}, a geometric characterization of sum-rank metric codes  has been provided. 
This is very useful to the aim of studying metric properties of sum-rank metric codes. We will indeed use this geometric characterization later with the principal goal of constructing optimal codes in this metric.
Let start by introducing the systems, that are the geometric counterpart of sum-rank metric codes.

\begin{definition}
An \textbf{$\Fmnkd$ system} $\mathcal{U}$ is an ordered set $(\mathcal{U}_1,\ldots,\mathcal{U}_t)$, where $\mathcal{U}_i$ is an $\F_q$-subspace of $\F_{q^m}^k$ of dimension $n_i$, for any $i\in [t]$, such that
$ \langle \mathcal{U}_1, \ldots, \mathcal{U}_t \rangle_{\F_{q^m}}=\F_{q^m}^k$ and $d$ is defined as
$$ d\coloneqq N-\max\left\{\sum_{i=1}^t\dim_{\F_q}(\mathcal{U}_i\cap \mathcal H) \st \mathcal H \textnormal{ is an $\F_{q^m}$-hyperplane of }\F_{q^m}^k\right\}.$$
\end{definition}

We say that two $\Fmnkd$ systems $(\mathcal{U}_1,\ldots,\mathcal{U}_t)$ and $(\mathcal{V}_1,\ldots, \mathcal{V}_t)$ are \textbf{equivalent} if there exists a matrix $A\in\GL(k,\F_{q^m})$, an element $\mathbf{a}=(a_1,\ldots,a_t)\in (\Fm^*)^t$ and a permutation $\sigma\in\mathcal{S}_t$, such that for every $i\in[t]$
$$ \mathcal{U}_i \cdot A = a_i\mathcal{V}_{\sigma(i)}.$$
We denote the set of equivalence classes of $\Fmnkd$ systems by $\mathfrak{U}\Fmnkd$.

\medskip 
Moreover, we recall this fundamental result. 
\begin{theorem} [see \textnormal{\cite[Theorem 3.1]{neri2023geometry}}] \label{th:connection}
Let $\C$ be a non-degenerate $\Fmnkd$ code with generator matrix $G=(G_1|\ldots|G_t)$.
Let $\mathcal{U}_i$ be the $\F_q$-span of the columns of $G_i$, for $i\in [t]$. Then, for every $v\in \Fm^k$ and $i \in [t]$ we have
\[
\ww(v G) = N - \sum_{i=1}^t \dim_{\fq}(\mathcal{U}_i \cap v^{\perp}).
\]
In particular,
$$ d= N-\max\left\{\sum_{i=1}^t\dim_{\F_q}(\mathcal{U}_i\cap \mathcal H) \st \mathcal H \textnormal{ is an $\F_{q^m}$-hyperplane of }\F_{q^m}^k\right\}.$$
\end{theorem}

We are ready to describe the  1-to-1  correspondence  between  equivalence  classes  of  sum-rank nondegenerate $\Fmnkd$ code and equivalence classes of $\Fmnkd$ systems.

Let $[\C]\in\mathfrak{C}\Fmnkd$ where $\C$ is a nondegenerate $\Fmnk$ code. Let $G=(G_1 \,\mid\, \ldots \,\mid\, G_t)$ be a generator matrix for $\cC$ and consider for each $i\in[t]$, $\mathcal{U}_i$ to be the $\Fq$-span of the columns of $G_i$. Define $\Psi([\C])$ as the equivalence class of $\Fmnkd$ systems $[\mathcal{U}]$, where $\mathcal{U}=(\mathcal{U}_1,\ldots,\mathcal{U}_t)$. In this case $\mathcal{U}$ is also called a \textbf{system associated with} $\cC$. Viceversa, let $[(\mathcal{U}_1,\ldots,\mathcal{U}_t)]\in\mathfrak{U}\Fmnkd$. For every $i\in[t]$, fix an $\F_q$-basis $\{g_{i,1}, \ldots, g_{i,n_i}\}$ of $\mathcal{U}_i$. Define $G_i$ as the matrix whose columns are $\{g_{i,1}, \ldots, g_{i,n_i}\}$ and let $\Phi([(\mathcal{U}_1,\ldots,\mathcal{U}_t)])$ be the equivalence class of the sum-rank metric code generated by $G=(G_1 \,\mid\, \ldots \,\mid\, G_t)$. 

Taking into account the notion of equivalence for $\Fmnkd$ codes and $\Fmnkd$ system and Theorem \ref{th:connection}, the following maps
\begin{align*}
    \Psi :  \mathfrak{C}\Fmnkd &\to\mathfrak{U}\Fmnkd \\
    \Phi : \mathfrak{U}\Fmnkd &\to \mathfrak{C}\Fmnkd,
\end{align*}
are defined.

\begin{theorem} [see \textnormal{\cite[Theorem 3.7]{neri2023geometry}}] 
$\Psi$ and $\Phi$ are well-defined and are the inverse of each other.
\end{theorem}

\medskip

\begin{remark} [see \textnormal{\cite[Remark 3.9]{neri2023geometry}}]
Note that when $t=1$, the notion of $\Fmnkd$ system coincides with the notion of $[n,k,d]_{q^m/q}$ system introduced in \cite{Randrianarisoa2020ageometric}; see also \cite{alfarano2022linear}. In the case where $n_1=\ldots=n_t=1$, this notion coincides with the notion of projective systems. Consequently, the above geometric characterization extend in the sum-rank metric the well-known one-to-one correspondence between equivalence classes of Hamming-metric nondegenerate codes and equivalence classes of projective systems; see \cite[Theorem 1.1.6]{vladut2007algebraic} and the one-to-one correspondence between classes of nondegenerate rank-metric codes and equivalence classes of $\Fmkd$-systems, estabilished in \cite{sheekeysurvey,Randrianarisoa2020ageometric}, see also \cite{alfarano2022linear}.
\end{remark}

Finally, we recall a duality introduced in \cite{santonastaso2022subspace} for $\Fmnkd$ systems by using their connection with sum-rank metric codes.

Let $\mathcal{U}=(\mathcal{U}_1,\ldots,\mathcal{U}_t)$ be an $\Fmnkd$ system and suppose that $d \geq 2$. Let $[\cC]=\Phi([\mathcal{U}])$. Since $d \geq 2$, by Proposition \ref{prop:charact_nondegenarate}, the dual codes of the codes in $[\cC]$ are nondegenerate and taking into account that two sum-rank metric codes $\C_1$ and $\C_2$ are equivalent if and only if their dual codes $\C_1^{\perp}$ and $\C_2^{\perp}$ are equivalent, we can consider the equivalence class $[\cC^{\perp}]$ of $[\bfn,N-k]_{q^m/q}$ codes that are the dual of the codes in $[\cC]$. We call $[\mathcal{U}^{\perp}]=\Psi([\cC^{\perp}])$ the \textbf{Delsarte class dual} of $[\mathcal{U}]$ and any $[\bfn,N-k]_{q^m/q}$ system in $[\mathcal{U}^{\perp}]$ a \textbf{Delsarte dual system} of $\mathcal{U}$.

\subsection{Matrix sum-rank metric codes and weight distributions of MSRD codes}
We can deal with sum-rank metric codes also in a matrix setting.
We consider the direct sum of $t$ matrix spaces
\[
\Mat(\bfn,m,\F_q):=\bigoplus_{i=1}^t\F_q^{n_i \times m}.
\]
We define the \textbf{sum-rank weight} of an element $X=(X_1|\ldots|X_t) \in \Mat(\bfn,m,\F_q)$ as 
$$w(X):=\sum_{i=1}^t \mathrm{rk}(X_i).$$
Consequently, the \textbf{sum-rank distance} is defined as 
\[
d(X,Y) := w(X-Y)= \sum_{i=1}^t \mathrm{rk}(X_i - Y_i),
\]
for any $X=(X_1 \lvert \ldots \lvert X_t)$, $Y=(Y_1 \lvert \ldots \lvert Y_t) \in \Mat(\bfn,m,\F_q)$, with $X_i,Y_i \in \F_q^{n_i \times m}$. 

A \textbf{(matrix) sum-rank metric code} $\mathrm{C}$ is an $\F_q$-linear subspace of $\Mat(\bfn,m,\F_q)$ endowed with the sum-rank distance and its minimum sum-rank distance is defined as $$d(\mathrm{C})=\min\{w(X): X \in \mathrm{C}, X \neq 0\}.$$ The matrix and vector settings for the sum-rank metric described above are related as following.
Let $\Gamma_r=(\gamma_1^{(r)},\ldots,\gamma_m^{(r)})$ be an ordered $\fq$-basis of $\F_{q^m}$, for each $r \in \{1,\ldots,t\}$ and define $\Gamma=(\Gamma_1,\ldots,\Gamma_t)$. For an element $x=(x_1, \ldots ,x_t) \in \F_{q^m}^\bfn$, with $x_i \in \F_{q^m}^{n_i}$, define the element $\Gamma(x)=(\Gamma_1(x_1), \ldots, \Gamma_t(x_t)) \in \Mat(\bfn,m,\F_q)$, where
$$(x_{r})_i = \sum_{j=1}^m (\Gamma_r (x_r))_{ij}\gamma_j^{(r)}, \qquad \mbox{ for each } i \in \{1,\ldots,n_r\}.$$
Equivalently, the $r$-th block of $\Gamma(x)$ is the matrix expansion of the vector $x_r$ with respect to the $\fq$-basis $\Gamma_r$ of $\F_{q^m}$.
As already stated in \cite[Theorem 2.7]{neri2023geometry}, the map

    \begin{equation} \label{eq:isometrymatrixvector}
    \Gamma: \F_{q^m}^\bfn \longrightarrow \Mat(\bfn,m,\F_q)
    \end{equation}
    
is an $\fq$-linear isometry between the metric spaces $(\F_{q^m}^\bfn, \mathrm{d})$ and $(\Mat(\bfn,m,\F_q),d)$. Also in the matrix framework, we can deal with optimal codes. Indeed, for a sum-rank metric code in $\Mat(\bfn,m,\F_q)$ having dimension $k$ and with minimum distance $d$, a Singleton-like bound holds and if its parameters attains the bound in Theorem \ref{th:Singletonbound}, it is called \textbf{Maximum Sum-Rank Distance code} (or shortly \textbf{MSRD code}) (see \cite[Theorem III.2]{byrne2021fundamental}). Clearly, if $\C$ is an MSRD $\Fmnkd$ code then $\Gamma(\C)$ is an MSRD code in $\Mat(\bfn,m,\F_q)$, as well.
     \begin{remark}
Actually, we can consider a more general ambient space for the matrix setting, i.e. the direct sum of $t$ matrix spaces $\F_q^{n_i \times m_i}$ in which not all the $m_i$'s are necessarily equal. But for our aim, we will restrict the study to $(\Mat(\bfn,m,\F_q),d)$. 
  \end{remark}
For more details on the matrix setting and its connection with the vector setting, see \cite{byrne2021fundamental} and \cite{neri2022twisted}. \\

\subsection{Linearized and Twisted linearized Reed-solomon codes}
The first construction of MSRD codes was introduced in \cite{Martinez2018skew} and it is known as \emph{linearized Reed-Solomon codes}. These codes are the analogue in the sum-rank metric of Gabidulin codes in the rank metric and the Reed-Solomon codes in the Hamming metric. See also \cite{neri2022twisted}, for a description of these codes in a skew polynomial setting.
\\

In order to describe such codes, we recall the notion of skew polynomials, which were originally presented in the seminal paper by Ore \cite{ore1933theory}, in a
more general setting.

Let $\sigma$ be a generator of $\Gal(\F_{q^m}/\F_q)$.
The ring of skew polynomials $\F_{q^m}[x;\sigma]$ is a set of formal polynomials
\[
\left\{f=\sum_{i=0}^d f_i x^i \st f_i \in \F_{q^m}, d \in \Z_{\geq 0} \right\}
\]
equipped with the ordinary addition 
$
\sum_{i \geq 0}f_ix^i+\sum_{i \geq 0}g_ix^i=\sum_{i \geq 0}(f_i+g_i)x^i
$
and the multiplication rule 
\[
x \cdot a=\sigma(a) \cdot x, \mbox{ for any } a \in \F_{q^m},
\]
extended to polynomials by associativity and distributivity.

 We say that the degree of a nonzero skew polynomial $\sum_{i=0}^d f_i x^i$ is $d$ if $f_d \neq 0$ and we assume that the zero polynomials has degree $-\infty$. For an integer $v$, we denote by $\F_{q^m}[x; \sigma]_{<v}$ the set of skew polynomials of degree smaller than $v$.

We recall the \emph{generalized operator evaluation} introduced in \cite{leroy1995pseudolinear}.
Let $f=\sum_{i=0}^d f_i x^i \in \F_{q^m}[x; \sigma]$ and $\mu \in \F_{q^m}$. Define the map
\[
\begin{array}{rccl}
f(\cdot)_{\mu}: & \F_{q^m}  & \longrightarrow & \F_{q^m} \\
& \beta & \longmapsto & \sum\limits_{i =0}^d f_i \sigma^i(\beta)\N_i(\mu),
\end{array} 
\]
where $\N_i(\mu)\coloneqq\prod_{j=0}^{i-1}\sigma^j(\mu)$. 

For $\boldsymbol{\mu}=(\mu_1,\ldots,\mu_t) \in \F_{q^m}^t$ and $\boldsymbol{\beta}=(\beta_1,\ldots,\beta_n) \in \F_{q^m}^n$, we can define the multi-point evaluation map as 

\[
\begin{array}{rccl}
\mathrm{ev}_{\muu,\boldsymbol{\beta}}: & \F_{q^m}[x;\sigma] & \longrightarrow & \F_{q^m}^\bfn \\
&f  & \longmapsto &  (f(\beta_1)_{\mu_1},\ldots,f(\beta_n)_{\mu_1} \lvert   \cdots \lvert  f(\beta_1)_{\mu_t},\ldots, f(\beta_n)_{\mu_t}).
\end{array} 
\]

The properties of the map $\mathrm{ev}_{\muu,\boldsymbol{\beta}}$ have been studied in \cite{Martinez2018skew} and they lead to the construction of the the following class of MSRD codes.

\begin{definition}[see \textnormal{\cite[Definition 31]{Martinez2018skew}}]\label{def:linRScodes}
Let $n \leq m$, $t \leq q-1$ and $1 \leq k \leq tn$. Let $\boldsymbol{\mu}=(\mu_1,\ldots,\mu_t) \in \F_{q^m}^t$ and $\boldsymbol{\beta}=(\beta_1,\ldots,\beta_n) \in \F_{q^m}^n$, such that the elements of $\boldsymbol{\beta}$ are linearly independent over $\F_q$ and assume that $\N_{q^m/q}(\mu_i)\neq \N_{q^m/q}(\mu_j)$, if $i \neq j$. \\
The $[(n,\ldots,n),k]_{q^m/q}$ code
\[
\C_{k,\muu,\boldsymbol{\beta}}=\{ \mathrm{ev}_{\muu,\boldsymbol{\beta}}(f) \st f \in \F_{q^m}[x;\sigma]_{<k} \}, 
\]
is called \textbf{linearized Reed-Solomon code}.
\end{definition}

In \cite[Theorem 4]{Martinez2018skew}, it has been shown that the linearized Reed-Solomon codes $\mathcal{C}_{k,\muu,\boldsymbol{\beta}}$ are MSRD codes.

Recently, in \cite{neri2022twisted} a new family of MSRD codes has been introduced.

\begin{definition} 
Let $n \leq m$ and $1 \leq k \leq tn$. Let $\boldsymbol{\mu}=(\mu_1,\ldots,\mu_t) \in \F_{q^m}^t$ and $\boldsymbol{\beta}=(\beta_1,\ldots,\beta_n) \in \F_{q^m}^n$, such that the elements of $\boldsymbol{\beta}$ are linearly independent over $\F_q$ and let $\N_{q^m/q}(\mu_i) = \lambda_i$ and suppose $\lambda_i \neq \lambda_j$, if $i \neq j$. Let $\eta \in \F_{q^m}$ with the property that $\N_{q^m/q}(\eta)(-1)^{km} \notin G$, with $G$ the multiplicative subgroup of $\F_{q}^*$ generated by $\{\lambda_1,\ldots,\lambda_t\}$.
The $[(n,\ldots,n),k]_{q^m/q}$ code
\[
\C_{k,\muu,\boldsymbol{\beta}}(\eta)=\{ \mathrm{ev}_{\muu,\boldsymbol{\beta}}(f) \st f(x)=f_0+\ldots+f_{k-1}x^{k-1}+\eta f_0 X^k, f_i \in \F_{q^m} \}, 
\]
is called \textbf{twisted linearized Reed-Solomon code}.
\end{definition}

The twisted linearized Reed-Solomon code can be considered as an extension in the sum-rank metric of twisted Gabidulin code in the rank metric \cite{sheekey2016new, lunardon2018generalized} and the twisted Reed-Solomon code in the Hamming metric \cite{beelen2017twisted}. In \cite[Theorem 6.3]{neri2022twisted}, it has been shown that such codes are MSRD codes. Note that when $\eta=0$, these codes coincide with the linearized Reed-Solomon codes, i.e. $\C_{k,\muu,\boldsymbol{\beta}}(0)=\C_{k,\muu,\boldsymbol{\beta}}$.

\subsection{Weight distribution of MSRD codes} In \cite{byrne2021fundamental} the weight distribution of MSRD codes is also determined, as we will describe in the following. Assume that $n=n_1=\ldots=n_t$ and let $m'=\max\{m,n\}$ and $n'=\min\{m,n\}$. For $i=\{1,\ldots,t\}$, let $\mathcal{L}=\mathcal{L}'\times \ldots \times \mathcal{L}'$ be the product lattice, where $\mathcal{L}'$ is the lattice of $\F_q$-subspaces of $\F_{q}^{n'}$. This means that an element of $\mathcal{L}$ is of the form $\mathbf{W}=(W_1,\ldots,W_t)$, where $W_i$ is an $\F_q$-subspace of $\F_{q}^{n'}$, for each $i$ and define $\dim_{\F_q}(\mathbf{W})=(\dim_{\F_q}(W_1), \ldots, \dim_{\F_q}(W_t))$. The rank function on $\mathcal{L}$ is defined by $rk_{\mathcal{L}}(\mathbf{W})=\sum_{i=1}^t \dim_{\F_q}(W_i)$ , for all $\mathbf{W}=(W_1,\ldots,W_t) \in \mathcal{L}$. Moreover,  we can define the partial order $\leq$ on $\mathbb{N}_0^t$ by
\[
\mathbf{v} \leq \mathbf{u} \Longleftrightarrow v_i \leq u_i, \mbox{for each }i \in \{1,\ldots,t\}
\]
for any $\mathbf{u}=(u_1,\ldots,u_t),\mathbf{v}=(v_1,\ldots,v_t) \in \mathbb{N}_0^t$. Moreover, define
\[
f_{\ell}(\mathbf{u})=\sum\limits_{\substack{\mathbf{v} \leq \mathbf{u}\\ \lvert \mathbf{v} \rvert =\ell}} \prod_{i=1}^t(-1)^{u_i-v_i}q^{\binom{u_i-v_i}{2}} \qbin{u_i}{v_i}{q},
\]
where $\mathbf{u}=(u_1,\ldots,u_t)$, $\ell \in \mathbb{N}_0$ and $\lvert \mathbf{v} \rvert:=\sum_{i=1}^t v_i$ if  $\mathbf{v}=(v_1,\ldots,v_t)$.

For a sum-rank metric code $\mathrm{C}$ in $\Mat(\bfn,m,\F_q)$, let denote by $W_r(\mathrm{C})$ the number of codewords having sum-rank weight $r$.
\begin{theorem} [see \textnormal{\cite[Theorem VI.5, Remark VI.11.]{byrne2021fundamental}}] \label{th:weightdistrMSRD}
Using the above notation, if $\mathrm{C}$ is an MSRD code in $\Mat(\bfn,m,\F_q)$ having minimum distance $d$, we have
    \[
    W_r(\mathrm{C})= \sum_{\mathbf{W} \in \mathcal{L}, rk_{\mathcal{L}}(\mathbf{W})=r} \sum_{\ell=d}^{r}(q^{m'(\ell-d+1)}-1)f_{\ell}(\dim_{\F_q}(\mathbf{W})).
    \]
\end{theorem}

 By the above theorem, we note that all MSRD codes with the same parameters have the same sum-rank distribution. For this reason, we will also denote $W_r(\mathrm{C}) $ by 
 \begin{equation} \label{eq:numberweightrMSRD}
     W_r(n',m',t):=\sum_{\mathbf{W} \in \mathcal{L}, rk_{\mathcal{L}}(\mathbf{W})=r} \sum_{\ell=d}^{r}(q^{m'(\ell-d+1)}-1)f_{\ell}(\dim_{\F_q}(\mathbf{W})),
\end{equation}
 if $\mathrm{C}$ is an MSRD code in $\Mat(\bfn,m,\F_q)$, with $n=n_1=\ldots=n_t$ and $m'=\max\{m,n\}$, $n'=\min\{m,n\}$.

We conclude this section with some results that will be useful in the next section. In particular we prove that every MSRD sum-rank metric code in $\Mat(\bfn,m,\F_q)$ with $n=n_1=\ldots=n_t$, having minimum distance $d$ and $t\leq q-1$, has at least one codeword of weight $j$, for any $d \leq j \leq N$. We start with an auxiliary lemma.
 
\begin{lemma} \label{lem:weightdistrreedsolomon}
    Let $\C_{k,\muu,\boldsymbol{\beta}}$ be an $\Fmnkd$ code as in Definition \ref{def:linRScodes}, with $d=tn-k+1$. Then for each $k \leq tn$ and any $r \in \{d,\ldots,tn\}$, there exists an element $x \in \C_{k,\muu,\boldsymbol{\beta}}$ such that
$ \ww(x)=r$.
\end{lemma}

\begin{proof}
     We proceed by induction on $k$. If $k=1$, all the elements in $\C_{k,\muu,\boldsymbol{\beta}}$ have sum-rank weight $tn$. So assume that the statement holds for $\C_{k,\muu,\boldsymbol{\beta}}$. Since $\C_{k,\muu,\boldsymbol{\beta}} \subseteq \C_{k+1,\muu,\boldsymbol{\beta}}$, there exists at least one element of sum-rank weight $r$ in $\C_{k+1,\muu,\boldsymbol{\beta}}$, for each $r \in \{tn-k+1,\ldots,tn-1,tn\}$. On the other side, $\C_{k+1,\muu,\boldsymbol{\beta}}$ is an MSRD code that implies that there exist at least one of its element having sum-rank weight $tn-k$. Hence, the statement holds also for $\C_{k+1,\muu,\boldsymbol{\beta}}$ that concludes our proof.
\end{proof}

\begin{proposition} \label{prop:numberweightMSRD}
    Assume that $n:=n_1=\ldots=n_t$, $t\leq q-1$ and $m \geq n$. Let $\mathrm{C}$ be an MSRD code in $\Mat(\bfn,m,\F_q)$ having minimum distance $d$, with $n=n_1=\ldots=n_t$.  For any $0 \leq i \leq tn-d$, there always exists at least one codeword $X \in \mathrm{C}$ having sum-rank weight $d+i$, i.e.
    \[
    W_r(n,m,t)>0,
    \]
    for each $d \leq r \leq tn$.  
\end{proposition}
\begin{proof}
Let consider the $\Fmnk$ code $\C_{k,\muu,\boldsymbol{\beta}}$ as in Definition \ref{def:linRScodes}. Let $\mathrm{C}_{k,\muu,\boldsymbol{\beta}}= \Gamma(\C_{k,\muu,\boldsymbol{\beta}})$, where $\Gamma$ is as in \eqref{eq:isometrymatrixvector}. Since $\C_{k,\muu,\boldsymbol{\beta}}$ is an MSRD code then $\mathrm{C}_{k,\muu,\boldsymbol{\beta}}$ is an MSRD code in $\Mat(\bfn,m,\F_q)$ as well. By Theorem \ref{th:weightdistrMSRD}, we know that all MSRD codes with the same parameters have the same sum-rank distribution. This allow us to look only to $\mathrm{C}_{k,\muu,\boldsymbol{\beta}}$. Then by Lemma \ref{lem:weightdistrreedsolomon}, we get the assertion.
\end{proof}

For the vectorial setting, we have the following.

\begin{corollary} \label{cor:numberweightMSRDvector}
    Assume that $n=n_1=\ldots=n_t$ and $t\leq q-1$. Let $\C$ be an MSRD $\Fmnkd$ code.  For any $0 \leq i \leq t\min\{m,n\}-d$, we have $W_{d+i}(\C)>0$.
\end{corollary}
\begin{proof}
 Let $\mathrm{C}= \Gamma(\C)$, where $\Gamma$ is as in \eqref{eq:isometrymatrixvector}. Since $\C$ is an MSRD code then also $\mathrm{C}$ is an MSRD code as well. Moreover,
 \[
 W_r(\C)=W_r(\mathrm{C})=W_r(n',m',t),
 \]
 for each $r$, where $m'=\max\{m,n\}$, $n'=\min\{m,n\}$. By Proposition \ref{prop:numberweightMSRD}, we have that \[W_r(n',m',t)>0,\] from which it follows the assertion. 
\end{proof}

\section{MSRD codes and subspace designs}\label{sec:designs}


In this section, we will use the geometric point of view described in the above section to give a characterization of the generalized weights of a sum-rank metric code via the associated system. Consequently, we will use relations on generalized sum-rank weights to estabilish 
a correspondence between MSRD codes and systems having special properties of intersection with respect to subspaces of fixed dimension, i.e. $s$-designs.

\subsection{Subspace designs}

The notion of subspace design was first given by Guruswami and Xing in \cite{guruswami2016explicit}. Subspace designs with good parameters were very useful in order to construct the first explicit families of rank metric codes and subspace codes which are efficiently list-decodable, see also \cite{guruswami2011linear}. Moreover, these objects were later used also to obtain
algebraic-geometric codes efficiently list-decodable \cite{guruswami2017optimal} and to give explicit constant degree dimension expanders over large fields \cite{guruswami2021lossless}.

\begin{definition}
An ordered set $\mathcal{U}=(\U_1,\U_2,\ldots,\U_t)$, where $\U_i$ is an $\F_q$-subspace of $V=V(k,q^m)$, for any $i \in \{1,\ldots,t\}$, is called a \textbf{(non-degenerate) $(h,r)$-subspace design} in $V$  if \newline $\dim_{\F_{q^m}} \langle \U_1,\ldots,\U_t \rangle_{\F_{q^m}} =V$ and  
$$\sum_{i=1}^t \dim_{\fq}(\U_i \cap W) \leq r,$$
for every $\F_{q^m}$-subspace $W \subseteq V$ of dimension $h$.
\end{definition}

Note that, when $t = 1$, the notion of $(s, A)$-subspace designs coincide with that of $(s, A)$-evasive subspaces. Evasive subspaces can be considered as the $q$-analogues of evasive sets originally introduced in \cite{pudlak2004pseudorandom} by Pudl\'ak and R\"{o}dl in order to
construct explicit Ramsey graphs. Such subpaces with good parameters were also later constructed and studied in \cite{dvir2012subspace} and \cite{bartoli2021evasive} and applied in several contexts, see e.g. \cite{guruswami2021lossless,guruswami2016explicit}.

In the next, we recall some main properties of a subspace design.

\begin{proposition} [see \textnormal{\cite[Proposition 1.3, Proposition 4.1]{santonastaso2022subspace}}]
\label{prop:diminuzione} 
If $(\U_1,\ldots,\U_t)$ is an $(h,r)$-subspace design in $V=V(k,q^m)$, then
\begin{enumerate}
    \item $r \geq h$;
    \item $(\U_1,\ldots,\U_t)$ is also an $(h-h',r-h')$-subspace design for any $h'\in \{0,\ldots,h-1\}$.
\end{enumerate}
\end{proposition}

When $r$ is the smallest possible, i.e. $r=h$, an $(h,h)$-subspace design in $V$ is also called to be an $h$-\textbf{design} in $V$. These objects were introduced and studied in \cite{santonastaso2022subspace}.

 When $V=V(k,q^m)=\F_{q^m}^k$ and $\U=(\U_1,\ldots,\U_t)$ is $(h,r)$-subspace design (or $h$-design) in $V$, such that $\dim_{\fq}(\U_i)=n_i$, we also write that $\U$ is an $(h,r)$-subspace design (or $h$-design) $[\bfn,k]_{q^m/q}$ system.

In the case that $t = 1$, the notion of $h$-design coincide with the notion of \textbf{$h$-scattered subspace}.
Scattered subspaces were first introduced by Blokhuis and Lavrauw in \cite{blokhuis2000scattered} for
$h =1$, then generalized independently by Lunardon in \cite{lunardon2017mrd} and by Sheekey and Van de Voorde in \cite{sheekeyVdV} for $h = k - 1$, and then by Csajb{\'o}k, Marino, Polverino and Zullo in \cite{csajbok2021generalising} for general $h$. These subspaces have attracted a lot of attention especially because their connection with maximum
rank distance codes, see \cite{polverino2020connections,marino2022evasive,sheekey2016new,zini2021scattered}.

\begin{remark} \label{rk:hdesimpliesh'}
By Proposition \ref{prop:diminuzione}, it follows immediately that if $(\U_1,\ldots,\U_t)$ is an $h$-design in $V(k,q^m)$ it is also an $i$-design in $V(k,q^m)$ for any $i \leq h$.
\end{remark}

Moreover, gluing together subspace designs we are able to get subspace designs in a bigger space as following.

\begin{proposition} [see \textnormal{\cite[Corollary 4.6]{santonastaso2022subspace}}]
\label{cor:directsum}
Let $V=V_1 \oplus \cdots \oplus V_\ell$ where $V_i=V(k_i,q^m)$ and $V=V(k,q^m)$. If $(\U_{i,1},\ldots,\U_{i,t})$ is an $(h,r_i)$-subspace design in $V_i$, for every $i\in \{1,\ldots,\ell\}$, then $(\U_1,\ldots,\U_t)$, where $\U_i=\U_{1,i} \oplus \cdots \oplus \U_{t,i}$, is an $(h,\sum_{i=1}^{\ell} r_i-(\ell-1)h)_q$-subspace design in $V$. 
\end{proposition}

Moreover, the following bound on the dimension of the elements of a subspace design holds.

\begin{theorem} [see \textnormal{\cite[Theorem 4.8]{santonastaso2022subspace}}] 
If $(\U_1,\ldots,\U_t)$ is an $h$-design in $V=V(k,q^m)$, then
\[ \dim_{\fq}(\U_i)\leq \begin{cases}
k, & \mbox{ if } m < h+1,\\
\frac{mk}{h+1}, & \mbox{ if } m \geq h+1,
\end{cases} \]
for every $i \in \{1,\ldots,t\}$.
\end{theorem}

\begin{definition} \cite[Definition 4.12]{santonastaso2022subspace}
Let $(\U_1,\ldots,\U_t)$ be an $h$-design in $V=V(k,q^m)$ such that $m\geq h+1$. If
\begin{equation*} 
 \dim_{\fq}(\U_i)= \frac{mk}{h+1},
 \end{equation*}
 for each $i \in \{1,\ldots,t\} $, then $(\U_1,\ldots,\U_t)$ will be called a \textbf{maximum $h$-design}.
\end{definition}



\subsection{Generalized weights}
Let $\C$ be an $\Fmnk$ code and $G=(G_1 \,|\, \ldots \,|\, G_t)\in \F_{q^m}^{k \times N}$ be an its generator matrix, with $G_i \in \Fm^{k \times n_i}$. Let denote by $(G_1 | \ldots  | G_t)_{j}$ the $j$-th column of $(G_1 | \ldots  | G_t)_{j}$ for $j=1,\ldots,N$. The elements of $\F_{q^m}^{\bfn}$ can be naturally seen as long vectors in $\F_{q^m}^N$. In this way, by considering $\C$ also as a subset in $\F_{q^m}^N$ and we can endow it with the Hamming metric. 

We recall that the \textbf{$r$-th generalized Hamming weight} of $\C$ is  defined as 
\begin{equation} \label{eq:generalizedhamming}
\begin{array}{rl} 
d_r^H(\C) & :=\min \left\{N- \left\lvert \{j: (G_1 | \ldots  | G_t)_{j} \in \mathcal{H} \}  \right\rvert \colon  \mathcal{H} \mbox{ is a } (k-r)\mbox{-dimensional } \F_{q^m}\mbox{-subspace of } \F_{q^m}^k \right\} \\
& = N- \max \left\{\left\lvert \{j: (G_1 | \ldots  | G_t)_{j} \in \mathcal{H} \}  \right\rvert \colon  \mathcal{H} \mbox{ is a } (k-r)\mbox{-dimensional } \F_{q^m}\mbox{-subspace of } \F_{q^m}^k \right\},
\end{array}
\end{equation}
for each $r \in \{1,\ldots,k\}$, see
\cite[Theorem 1.1.14]{vladut2007algebraic}.

The notion of generalized weight for sum-rank metric was introduced in \cite{martinez2019theory}.
Since the sum-rank metric extends both the Hamming and rank metrics, we can also see that generalized sum-rank weights specialize to
generalized Hamming weights \cite{helleseth1977weight,luo2005some,wei1991generalized} and generalized rank weights \cite{kurihara2015relative,oggier2012existence} when $n_i = 1$
and $t = 1$, respectively.
In \cite[Section 4]{martinez2019theory}, several characterization of generalized sum-rank weight are given, and in this paper we will assume as definition one of these characterization \cite[Proposition 4]{martinez2019theory} that is more useful for our aims.

\begin{definition}
Let $\C$ be a $\Fmnkd$ code. For each $r \in \{1,\ldots,k \}$, the non negative integer
\[
d_r^{srk}(\C):= \min \{d^H_r(\C \cdot A) \colon A=diag(A_1,\ldots,A_t), A_i \in \GL(n_i,q), \mbox{ for } i=1,\ldots,t  \} 
\]
is called the \textbf{$r$-th generalized sum-rank weight} of $\C$.
\end{definition}

Using the geometric description of $\Fmnkd$, we are able to give a geometric characterization of generalized sum-rank weights as following.

\begin{proposition} \label{prop:geogeneralizedweights}
Let $\C$ be a nondegenerate $\Fmnkd$ code and let $\mathcal{U}=(\mathcal{U}_1,\ldots, \mathcal{U}_t) \in \Psi([\C])$. It holds that
\[
\ds_r(\C)= N- \max \left\{ \sum_{i=1}^t \dim_{\F_q}(U_i \cap \mathcal{H}) \colon \mathcal{H} \mbox{ is a } (k-r)\mbox{-dimensional } \F_{q^m}\mbox{-subspace of } \F_{q^m}^k \right\}. 
\]

\end{proposition} 
\begin{proof}
Let $G=(G_1 | \ldots |G_t)$ be a generator matrix for $\C$ such that the $\F_q$-span of the columns of $G_i$ is $\mathcal{U}_i$. If $A=diag(A_1,\ldots,A_t),$ with $ A_i \in \GL(n_i,q), \mbox{ for } i=1,\ldots,t$, then a generator matrix for $ \C \cdot A$ is the matrix $(G_1A_1 | \ldots | G_t A_t)$. So, using \eqref{eq:generalizedhamming}, we have that
\[
\begin{array}{rl}

\ds_r(\C) = &  N-\max \{ 
\lvert \{j:   (G_1A_1 | \ldots  | G_t A_t)_{j} \in \mathcal{H} \} \rvert  \colon  A_i \in \GL(n_i,q), 
\\ & \hskip 2 cm \mbox{and } \mathcal{H} \mbox{ is a }(k-r)\mbox{-dimensional } \F_{q^m}\mbox{-subspace of } \F_{q^m}^k \}
 \\
= &  N-\max \left\{ 
 \sum\limits_{i=1}^t \lvert \{j:   (G_iA_i)_{j} \in \mathcal{H} \} \rvert  \colon  A_i \in \GL(n_i,q), 
 \right. \\ &  \hskip 2 cm \left. 
 \mbox{and } \mathcal{H} \mbox{ is a }(k-r)\mbox{-dimensional } \F_{q^m}\mbox{-subspace of } \F_{q^m}^k  \right\}
 \\
\end{array}
\]
By noting that the $G_iA_i$'s , when $A_i$ runs in $\GL(n_i,q)$, give all the matrices whose columns is an $\F_q$-basis of $\mathcal{U}_i$, we get the assertion.
\end{proof}

\begin{remark}
    When $t=1$, the above characterization of generalized sum-rank weights corresponds exactly to the one proved in \cite{Randrianarisoa2020ageometric} given for generalized rank weights.
\end{remark}

Finally, we recall important properties of generalized sum-rank weights that we will use later. We start by giving the monotonicity property of generalized sum-rank weights.

\begin{proposition} [see \textnormal{\cite[Lemma 4]{martinez2019theory}}] \label{prop:monotonweight}
    Let $\C$ be a nondegenerate $\Fmnkd$ code. Then 
    \[
    d=d_1^{srk}(\C)<d_2^{srk}(\C)<\ldots<d_k^{srk}(\C)=N.
    \]
\end{proposition}

 We also recall the Wei's duality theorem for sum-rank metric codes that establish a relation between generalized sum-rank weight of a codes and of its dual.

\begin{theorem} [see \textnormal{\cite[Theorem 2]{martinez2019theory}}] \label{th:weitype}
    Let $\C$ be a nondegenerate $\Fmnkd$ code. Then it holds
    \[
    \{1,2,\ldots,N\}=\{d_i^{srk}(\C) \colon 1 \leq i \leq k \} \cup \{N+1-d_i^{srk}(\C^{\perp}) \colon 1 \leq i \leq N-k \},
    \]
    where the union is disjoint and $N=\sum_{i=1}^t n_i$. In particular, the generalized sum-rank weights of $\C$ uniquely determine those of $\C^{\perp}$.
\end{theorem}

\subsection{Connection with MSRD codes}

From Theorem \ref{th:connection}, we can rephrase the MSRD property of a code as geometric properties of the system.

\begin{corollary} [see \textnormal{\cite[Corollary 3.10]{neri2023geometry}, \cite[Theorem 6.10]{santonastaso2022subspace}}]\label{teo:designMSRD} 
Let $\C$ be a nondegenerate $\Fmnkd$ code and let $[(\mathcal{U}_1,\ldots,\mathcal{U}_t)]=\Psi([\mathcal{C}])$. $\mathcal{C}$ is an MSRD code with minimum distance \[
	d = \begin{cases} N-k+1 &  \mbox{if }m \geq n_1
		 \\
 tm-\frac{mk}{n}+1 &\mbox{ if }m<n=n_1=\ldots=n_t
	\end{cases}
	\] if and only if 
\[
\max\left\{ \sum_{i=1}^t \dim_{\fq}(\mathcal{U}_i \cap \mathcal{H})  \st \mathcal{H} \mbox{ hyperplane of }\F_{q^m}^k  \right\} \leq \begin{cases}
k-1, & \mbox{ if } m \geq n_1, \\ \\
t(n-m)+\frac{m}{n}k-1, & \mbox{ if } m<n:=n_1=\ldots=n_t
\end{cases}.
\]
\end{corollary}

A subspace design in $\F_{q^m}^k$ which corresponds to an MSRD code is called \textbf{optimal subspace design}.
 As a consequence of Corollary \ref{teo:designMSRD}, in the case $m\geq n_1$, we have that MSRD codes correspond to $(k-1)$-design.
\begin{theorem} [see \textnormal{\cite[Corollary 3.10]{neri2023geometry}}]
    Let $\C$ be a nondegenerate $\Fmnkd$ code. Assume that $m \geq n_1$. Let $(\mathcal{U}_1,\ldots,\mathcal{U}_t) \in \Psi([\C])$. Then $\C$ is an MSRD code if and only if $(\mathcal{U}_1,\ldots,\mathcal{U}_t)$ is a $(k-1)$-design $[\bfn,k]_{q^m/q}$ system.
\end{theorem}

\begin{corollary}
    An $\Fmnkd$ system $\mathcal{U}$, with $m \geq n_1$ is an optimal subspace design if and only if it is a $(k-1)$-design in $\F_{q^m}^k$. In such a case, it is also a maximum $(k-1)$-design if $m=n_1=\ldots=n_t$.
\end{corollary}

In the next theorem, we investigate the case in which all the blocks of the code have the same length and this quantity is greater that $m$. To this aim, first we relate the generalized sum-rank weights of a code with geometric properties of a system associated to it, as recently done for the rank metric codes in \cite[Theorem 3.3]{marino2022evasive}.

\begin{theorem} \label{th:connectiondesignweight}
    Let $\mathcal{C}$ be a nondegenerate $\Fmnkd$ code with $k<N$ and let $\mathcal{U}$ be a $\Fmnkd$ system associated to $\cC$. Then, the following are equivalent
    \begin{enumerate}
        \item $\mathcal{U}$ is $(h,r)$-subspace design.
        \item $\ds_{k-h}(\mathcal{C}) \geq N-r$.
        \item $\ds_{r-h+1}(\mathcal{C}^{\perp}) \geq r+2$.
    \end{enumerate}
   
\end{theorem}
\begin{proof}
Using the characterization of generalized sum-rank weight as in Proposition \ref{prop:geogeneralizedweights} and using the Proposition \ref{prop:monotonweight} and Theorem \ref{th:weitype} for them, the
proof of can be translated mutatis mutandis from that in \cite[Theorem 3.3]{marino2022evasive} replacing the notion of evasive subspace with the one of subspace design.
\end{proof}

From the above result, we are able to give a characterization of the MSRD codes in terms of system.

\begin{theorem} \label{Th:sdesignMSRD}
Let $\C$ be a nondegenerate $\Fmnkd$ code. Assume that $n:=n_1=\ldots=n_t \geq m$. Let $(\mathcal{U}_1,\ldots,\mathcal{U}_t) \in \Psi([\C])$. Then $\C$ is an MSRD code if and only if $n=\frac{mk}{h+1}$ such that  $h+1$ divides $mk$, for some positive integer $h \leq k-1$ and $(\mathcal{U}_1,\ldots,\mathcal{U}_t)$ is an $h$-design $[\bfn,k]_{q^m/q}$ system. In particular if also $n \geq k$,  $\C$ is an MSRD code if and only if $(\mathcal{U}_1,\ldots,\mathcal{U}_t)$ is a maximum $h$-design $[\bfn,k]_{q^m/q}$ system.
\end{theorem}

\begin{proof}
    Suppose that $\C$ is an MSRD $\Fmnkd$ code. Then by Theorem \ref{th:Singletonbound}, we have that $mk=n(tm-d+1)$. This means that $n$ divides $mk$ and since $n \geq m$, it follows that $n=\frac{mk}{h+1}$, where $h=tm-d$ and $h \leq k-1$. Also, by Theorem \ref{th:dualMSRD}, $\C^{\perp}$ is an MSRD code, with $\ds_{1}(\C^{\perp})=d(\C^{\perp})=h+2$. Using Theorem \ref{th:connectiondesignweight}, we get that $(\mathcal{U}_1,\ldots,\mathcal{U}_t)$ is an $h$-design. If also $n \geq k$ holds, then $m \geq k+1$ and so $(\mathcal{U}_1,\ldots,\mathcal{U}_t)$ is a maximum $h$-design. \\
    Conversely, if $h$ is a positive integer such that $h+1$ divides $mk$, and $(\mathcal{U}_1,\ldots,\mathcal{U}_t)$ is an $h$-design $[\bfn,k]_{q^m/q}$ system, with $n=\frac{mk}{h+1}$, again by Theorem \ref{th:connectiondesignweight}, we have that $\C^{\perp}$ is an $[\bfn,N-k]_{q^m/q}$ code with $\ds_{1}(\C^{\perp})=d(\C^{\perp})\geq h+2$. By Theorem \ref{th:Singletonbound}, $d(\C^{\perp}) = h+2$ and so $\C^{\perp}$ is an MSRD code and as a consequence $\C$ is an MSRD code as well.
\end{proof}

\begin{corollary} \label{cor:optimalk-1design}
    An $\Fmnk$ system $\mathcal{U}$, with $n:=n_1=\ldots=n_t \geq m$ is an optimal subspace design if and only if it is a $h$-design in $\F_{q^m}^k$, $n=\frac{mk}{h+1}$ such that $h+1$ divides $mk$, for some positive integer $h \leq k-1$. In such a case, it is also a maximum $h$-design if $m \geq h+1$ (or equivalently $n \geq k$).
\end{corollary}

\begin{remark}
    Note that since $m \geq 2$ an optimal subspace design in $\F_{q^m}^k$ that is an $1$-design is always a maximum $1$-design in $\F_{q^m}^k$.
\end{remark}

From Theorem \ref{Th:sdesignMSRD}, we get a characterization of $h$-design using their property of the dimension of intersections of hyperplanes. This also improves the bounds proved in \cite[Corollary 5.11]{santonastaso2022subspace} and extend \cite[Theorem 7.3]{santonastaso2022subspace} for $h>1$.

\begin{corollary} \label{Cor:hdesignhyper}
Let $\U=(\mathcal{U}_1,\ldots,\mathcal{U}_t)$ be a $[\bfn,k]_{q^m/q}$ system, with $n:=n_1=\ldots=n_t=\frac{mk}{h+1}$ and $h \leq k-1$. Then $\U$ is an $h$-design in $\F_{q^m}^k$ if and only if
\[
t \left( n-m \right) \leq \sum_{i=1}^t \dim_{\F_q}(\mathcal{U}_i \cap H) \leq t \left(n-m \right)+h,
\]
for any $\F_{q^m}$-hyperplane $H$ of $\F_{q^m}^k$. In particular, if $\mathcal{U}$ is an $h$-design then it is also an \newline $\left(k-1, t \left(n-m \right)+h\right)$-subspace design in $\F_{q^m}^k$.
\end{corollary}

\begin{proof}
    By Theorem \ref{Th:sdesignMSRD}, we have that $(\mathcal{U}_1,\ldots,\mathcal{U}_t)$ is a $h$-design in $\F_{q^m}^k$ if and only if any code $\C \in \Psi([\mathcal{U}])$ is an MSRD $[\bfn,k]_{q^m/q}$ code. Then the assertion follows by Theorem \ref{teo:designMSRD}.
\end{proof}

{Corollary \ref{cor:numberweightMSRDvector} allows us to completely determine the properties of intersection of an $h$-design with respect to the hyperplanes, extending \cite[Theorem 5.12]{santonastaso2022subspace} to any $1 \leq h \leq k-1$.

\begin{theorem} \label{th:numberhyperplaneshdesign}
    Let $\U=(\mathcal{U}_1,\ldots,\mathcal{U}_t)$ be an $h$-design $[\bfn,k]_{q^m/q}$ system, with $n=\frac{mk}{h+1}$ and $h \leq k-1$. For any $j \in \{0,\ldots,h\}$, the number of hyperplanes of $\F_{q^m}^k$ such that
\begin{equation} \label{eq:intersectionhyperhdesign}
 \sum_{i=1}^t \dim_{\F_q}(\mathcal{U}_i \cap H) = t \left(n-m \right)+j,
\end{equation}
is equal to \[g_j=\frac{W_{tm-j}\left(m,\frac{mk}{h+1},t\right)}{q^m-1},\]
where $W_{tm-j}\left(m,\frac{mk}{h+1},t\right)$ is defined as in \eqref{eq:numberweightrMSRD}.
In particular, $g_j>0$, for each $j$.
\end{theorem}

\begin{proof}
Let $\C \in \Phi([\U])$. By Theorem \ref{Th:sdesignMSRD}, we know that $\C$ is an MSRD $\Fmnkd$ code with $d=tm-h$. Using Theorem \ref{th:connection}, we get that the number of hyperplanes satisfying \eqref{eq:intersectionhyperhdesign} is $\frac{W_r(\C)}{q^m-1}$. From Theorem \ref{th:weightdistrMSRD} and Corollary \ref{prop:numberweightMSRD}, we have the assertion.
\end{proof}

\subsection{Geometric dual of subspace design and sum-rank metric codes}

We can recall also another notion of dual of a subspace design, i.e. the geometric dual subspace design.
For a nondegenerate reflexive sesquilinear form $\sigma \colon V \times V \rightarrow \F_{q^m}$ on $V(k,q^m)$  we can consider \[
\begin{array}{cccc}
    \sigma': & V \times V & \longrightarrow & \F_q  \\
     & (x,y) & \longmapsto & \mathrm{Tr}_{q^m/q} (\sigma(x,y)).
\end{array}
\] 
It results that $\sigma'$ is a nondegenerate reflexive sesquilinear form on $V(k,q^m)=V(km,q)$ seen as an $\fq$-vector space of dimension $km$. This allows to consider $\perp$ and $\perp'$ as the orthogonal complement maps defined by $\sigma$ and $\sigma'$, respectively. Moreover, it can be proved that for any $\F_{q^m}$-subspace of $V$ holds $W^{\perp'}=W^{\perp}$.
This leads to the following property. For an $\F_{q^m}$-subspace $W$ of $V$ and an $\F_q$-subspace $U$ of $V$, we have 
\begin{equation}\label{eq:dualord}
\dim_{\fq}(U^{\perp'}\cap W^{\perp})=\dim_{\fq}(U\cap W)+\dim_{\fq}(V)-\dim_{\fq}(U)-\dim_{\fq}(W).
\end{equation}
For more details see \cite[section 7]{taylor1992geometry} and \cite[Property 2.6]{polverino2010linear}.

Let consider an ordered set $(\U_1,\ldots,\U_t)$ of $\F_q$-subspaces in $V=V(k,q^m)$.
The ordered set $(\U_1^{\perp'},\ldots,\U_t^{\perp'})$ will be called the \textbf{dual subspaces} (w.r.t. $\perp'$) of $(\U_1,\ldots,\U_t)$.
Furthermore, note that the dual subspaces of $(\U_1^{\perp'},\ldots,\U_t^{\perp'})$ coincide with the ordered set $(\U_1,\ldots,\U_t)$.

From Equation \eqref{eq:dualord}, one can get the following result.

\begin{theorem} [see \textnormal{\cite[Theorem 7.1]{santonastaso2022subspace}}] \label{th:orddual}
Suppose that $(\U_1,\ldots,\U_t)$ is an $(h,r)$-subspace design in $V=V(k,q^m)$, with $\dim_{\fq} (\U_i) =n_i$ for any $i$ and suppose that $\langle \U_1^{\perp'},\ldots,\U_t^{\perp'}\rangle_{\F_{q^m}}=V$.
Then the dual subspaces $(\U_1^{\perp'},\ldots,\U_t^{\perp'})$ of $(\U_1,\ldots,\U_t)$ is an $(k-h,r+t(k-h)m-\sum_{i=1}^t n_i)$-subspace design where $\dim_{\fq} (\U_i^{\perp'})=mk-n_i$. 
\end{theorem}

\begin{remark} \label{rk:geogenrank}
Note that if $(\U_1,\ldots,\U_t)$ is an $\Fmnk$ system then the dual subspaces $(\U_1^{\perp'},\ldots,\U_t^{\perp'})$ form an $[\bfn',k]_{q^m/q}$ system, where $\bfn'=(km-n_1,\ldots,km-n_t)$, if the condition $\langle \U_1^{\perp'},\ldots,\U_t^{\perp'}\rangle_{\F_{q^m}}=\F_{q^m}^k$ holds.
This condition is equivalent to the fact that $\langle \U_1^{\perp'},\ldots,\U_t^{\perp'}\rangle_{\F_{q^m}}$ is not contained in any hyperplane $H$ of $\F_{q^m}^k$, which dually corresponds to 
\[
H^{\perp} \not \subseteq \U_1 \cap \ldots \cap  \U_t,
\]
for every hyperplane $H$ of $\F_{q^m}^k$.
By Proposition \ref{prop:geogeneralizedweights}, in term of codes this is equivalent to require that \[d_{k-1}^{\mathrm{rk}}(\C)\geq N-tm+1.\]
For more details, see \cite[Remark 3.3]{borello2023geometric}.
\end{remark}

If $\U=(\U_1,\ldots,\U_t)$ is an $\Fmnk$ system such that $\U_1 \cap \ldots \cap \U_t $ does not contain any $1$-dimensional $\F_{q^m}$-subspace of $\F_{q^m}^k$, then we call the dual subspaces $(\U_1^{\perp'},\ldots,\U_t^{\perp'})$ of $\U$ \textbf{the geometric dual system} of $\U$.
The notion of geometric dual of a system leads to an operation on sum-rank metric codes called \textbf{geometric dual}, that to any element in $\mathfrak{C}[\bfn,k]_{q^m/q}$ gives another element in $\mathfrak{C}[\bfn',k]_{q^m/q}$, where $\bfn'=(km-n_1,\ldots,km-n_t)$, by preserving the dimension of the code. Taking this in mind, Remark \ref{rk:geogenrank}, we recall the following definition, see \cite{borello2023geometric}.
\begin{definition}
Let $\C$ be a non-degenerate $\Fmnk$ code and let $\U \in \Psi([\C])$.
Suppose also that $\mathrm{d}^{\mathrm{srk}}_{k-1} (\C)\geq N-tm+1$. Then a \textbf{geometric dual} $\C^{\perp_{\mathcal{G}}}$ of $\C$ (with respect to $\perp'$, where $\perp'$ is defined as above) is defined as a code $\C'$, where $\C'$ is any code associated with a dual system $\U^{\perp'}$ of $\U$.
\end{definition}

This definition is also justified by the following result.

\begin{theorem} [see \textnormal{\cite[Theorem 3.4]{borello2023geometric}}]
Let $\C$ be an $[\bfn,k,d]_{q^m/q}$ code such that $\ds_{k-1}(\C)\geq N-tm+1$, and let $\U \in \Psi([\C])$.
Then, up to equivalence, a geometric dual $\C^{\perp_{\mathcal{G}}}$ of $\C$ does not depend on the choice of the associated system and on the choice of code in $[\C]$, hence $\perp_{\mathcal{G}}$ is well-defined.
Moreover, $[\C^{\perp_{\mathcal{G}}}] \in \mathfrak{C}[\bfn',k]_{q^m/q}$, where $\bfn'=(km-n_1,\ldots,km-n_t)$.
\end{theorem}

First we can characterize $h$-design using geometric property of their geometric dual.

\begin{proposition} \label{prop:geometricdualhdesign}
    Let $\U=(\U_1,\ldots,\U_t)$ an ordered set of $\F_q$-subspaces of $V=V(k,q^m)$ such that $\dim_{\F_q}(\U_i)=\frac{mk}{h+1}$, for each $i$ and $\U_1 \cap \ldots \cap \U_t $ does not contain any $1$-dimensional $\F_{q^m}$-subspace of $\F_{q^m}^k$. Then $\U$ is an $h$-design if and only if the dual subspaces $(\U_1^{\perp'},\ldots,\U_t^{\perp'})$ of $\U$ form an $(1,h)$-subspace design such that $\dim_{\F_q}(\U_i^{\perp'})=\frac{hmk}{h+1}$, for each $i$.
\end{proposition}

\begin{proof}
    By Corollary \ref{Cor:hdesignhyper} $\U$ is an $h$-design if and only if it is an $\left(k-1, t \left(n-m \right)+h\right)$-subspace design. Then the assertion follows by Theorem \ref{th:orddual}.
\end{proof}

\begin{remark} \label{rk:dualstrictlyhdesign}
 Note that, when $\U_1,\ldots,\U_t$ is a maximum $h$-design and so the dual subspaces $(\U_1^{\perp'},\ldots,\U_t^{\perp'})$ of $\U$ form an $(1,h)$-subspace design such that $\dim_{\F_q}(\U_i^{\perp'})=\frac{hmk}{h+1}$, for each $i$ by Theorem \ref{th:numberhyperplaneshdesign}, we have also that, for any $j \in \{0,\ldots,h\}$, there exists at least $1$-dimension $\F_{q^m}$-subspace $\langle v \rangle_{\F_{q^m}} $such that $\sum_{i=1}^t \dim_{\F_q}(\U_i \cap \langle v \rangle_{\F_{q^m}})=j$. In particular, each of the $\U_i^{\perp'}$ is a $(\mathcal{D},h)$-scattered, where $\mathcal{D}$ is the Desarguesian spread; see \cite{gruica2022generalised}.
\end{remark}

In terms of code, we have a characterization of MSRD codes in terms of their geometric dual. First we observe that we can always deal with a geometric dual of an MSRD code.

\begin{lemma}
    Let $\C$ be a nondegenerate $\Fmnkd$ code. If $\C$ is an MSRD code then $\ds_{k-1}(\C) \geq N-1 \geq N-tm+1$. 
\end{lemma}

\begin{proof}
    The statement clearly holds if $m \geq n$. So, suppose that $n >m$. Since $\C$ is an MSRD code, by Theorem \ref{Th:sdesignMSRD}, we have that a system $(\mathcal{U}_1,\ldots,\mathcal{U}_t)$ associated with $\C$ is an $h$-design $[\bfn,k]_{q^m/q}$ system. This implies that $(\mathcal{U}_1,\ldots,\mathcal{U}_t)$ is also an $1$-design as well, see Remark \ref{rk:hdesimpliesh'}. Using Proposition \ref{prop:geogeneralizedweights}, this means that $d_{k-1}(\C) \geq N-1 \geq N-tm+1$.
\end{proof}

\begin{proposition} 
Let $\C$ be a nondegenerate $\Fmnkd$ code. Assume that $n \geq m$. Let $(\mathcal{U}_1,\ldots,\mathcal{U}_t) \in \Psi([\C])$. Then $\C$ is an MSRD code if and only if an its geometric dual $\C^{\perp_{\mathcal{G}}}$ is an $\left[\left(\frac{hmk}{h+1},\ldots,\frac{hmk}{h+1} \right),k \right]_{q^m/q}$ code such that $\ds_{k-1}(\C^{\perp_{\mathcal{G}}}) \geq N-h$.
\end{proposition}

\begin{proof}
    By Theorem \ref{Th:sdesignMSRD}, we know that $\C$ is an MSRD code if and only if $n=\frac{mk}{h+1}$, for some positive integer $h \leq k-1$ such that $h+1$ divides $mk$, and $(\mathcal{U}_1,\ldots,\mathcal{U}_t)$ is an $h$-design $[\bfn,k]_{q^m/q}$ system. So using Proposition \ref{prop:geometricdualhdesign} and taking into account that a system associated to a geometric dual is equivalent to the geometric dual of a system associated to $\C$, we get the assertion.
\end{proof}


Note that the geometric dual of an MSRD code is not in general an MSRD code, as evidenced by the following proposition.

\begin{proposition}
    Let $\C$ be a nondegenerate MSRD $\Fmnkd$ code. Assume that $n \geq m$, with $n=\frac{mk}{h+1}$ for some positive integer $h \leq k-1$ such that $h+1$ divides $mk$.  Then the geometric dual $\C^{\perp_{\mathcal{G}}}$ of $\C$ is an MSRD code if and only if $h=1$.
\end{proposition}
\begin{proof}
Let $(\mathcal{U}_1,\ldots,\mathcal{U}_t) \in \Psi([\C])$. By Theorem \ref{teo:designMSRD}, we know that $(\mathcal{U}_1,\ldots,\mathcal{U}_t)$ is a maximum $h$-design $[\bfn,k]_{q^m/q}$ system. Using Proposition \ref{prop:geometricdualhdesign}, we have that the dual subspaces $(\U_1^{\perp'},\ldots,\U_t^{\perp'})$ of $\U$ form an $(1,h)$-subspace design such that $\dim_{\F_q}(\U_i^{\perp'})=\frac{hmk}{h+1}$, for each $i$. So, if $h=1$, then $(\U_1^{\perp'},\ldots,\U_t^{\perp'})$ is a maximum $1$-design and so $\C^{\perp_{\mathcal{G}}}$ of $\C$ is an MSRD code, by Theorem \ref{Th:sdesignMSRD}. While, if $h>1$, then we know
 $\C^{\perp_{\mathcal{G}}}$ is an MSRD if   $(\U_1^{\perp'},\ldots,\U_t^{\perp'})$ is a maximum $h'$-design for some $h'$. In particular, by Proposition \ref{prop:diminuzione}, $(\U_1^{\perp'},\ldots,\U_t^{\perp'})$ would be also a $1$-design, a contradiction to Remark \ref{rk:dualstrictlyhdesign}. 
\end{proof}

\section{Disjoint maximum scattered linear sets and 
maximum 1-design}

In this section, we will provide constructions of maximum $1$-design $[\bfn,k]_{q^m/q}$ system, for any value of $m,k$ such that $mk$ is even. This will solve the problem the problem of existence for optimal sum-rank metric codes having at least two blocks and all the blocks of length $mk/2$.

\medskip 

\subsection{Scattered linear sets.} 


First, we recall the projective version of the system, i.e. linear sets.
Let $V$ be a $k$-dimensional $\F_{q^m}$-vector space. By $\Lambda=\PG(V,\F_{q^m})=\PG(k-1,q^m)$ we mean the projective space defined by $V$ over $\F_{q^m}$. For $\U $ an $\fq$-subspace of $V$ of dimension $n$, the set of points
\[ L_{\U}=\{\la {u} \ra_{\mathbb{F}_{q^m}} : {u}\in \U\setminus \{{ 0} \}\}\subseteq \Lambda \]
is called an $\fq$-\textbf{linear set of rank $n$}. Note that the rank of a linear set $L_{\U}$
is defined by the underlying subspace $\U$ that defines it. So, we have to keep in mind both the set of points defined by $U$ and the subspace $\U$.

The \textbf{weight} of a subspace $\mathcal{S}=\PG(W,\F_{q^m})\subseteq \Lambda$ in $L_{\U}$ is defined as 
\[ w_{L_U}(\mathcal{S})=\dim_{\fq}(\U\cap W). \]

For an $\F_q$-linear set $L_U$ of rank $n$ we have that
\[
    |L_{\U}| \leq \frac{q^n-1}{q-1}.
\]
When $|L_{\U}|$ satisfies the equality, $L_{\U}$ is called \textbf{1-scattered} (or simply \textbf{scattered}), or equivalently, if all the points of $L_{\U}$ have weight one, i.e. 
\[
w_{L_{\U}}(P) \leq 1,
\]
for every point $P \in \Lambda$.
Note that $L_{\U}$ is scattered if and only if $\U$ is $(1,1)_q$-evasive subspace in $V$, that is a $1$-scattered subspace without the assumption that it spans the entire space.
Finally, a canonical subgeometry is an $\fq$-linear set $L_{\U}$ of $\PG(k-1,q^m)=\PG(V,\F_{q^m})$ having rank $k$ and such that $\langle \U \rangle_{\F_{q^m}}=V$.
In particular, canonical subgeometries are scattered linear sets of rank $k$.
When $m=2$, canonical subgeometries are also called Baer subgeometries.

We refer to \cite{lavrauw2015field} and \cite{polverino2010linear} for comprehensive references on linear sets and their applications.

\medskip

Note that, for an $n$-dimensional 1-scattered subspace in $V(k,q^m)$, and for a scattered linear set of $\PG(k-1,q^m)$ having rank $n$, it holds that $n \leq km/2$. If the equality holds it is called \textbf{maximum 1-scattered subspace} and \textbf{maximum scattered linear set}, respectively.

The existence of maximum scattered linear sets has been proved in \cite{blokhuis2000scattered} for $k$ even, in \cite{bartoli2018maximum} for most of cases, and then the problem was completely solved in \cite{csajbok2017maximum} for the remaining cases.

Here we recall some constructions in the case that $k$ is odd.

\begin{theorem} [see \textnormal{\cite[Theorem 2.2]{bartoli2018maximum}}] \label{th:masxcattkodd1}
Let $\omega \in \F_{q^{2m}} \setminus \F_{q^{m}}$ and $a \in \F_{q^{km}}^*$. Assume $k \geq 3$ odd and define
\[
\U=\{ ax^{q^s}+x\omega \colon x \in  \F_{q^{km}}\} \subseteq \F_{q^{2km}},
\]
for some positive integer $s$.
If 
   \begin{itemize}
       \item $\gcd(m,k)$=1,
       \item $\gcd(s,2m)=1$ and $\gcd(s,km)=r$,
       \item $N_{q^{km}/q^r}(a) \notin \F_q$
   \end{itemize}

then $L_{\U}$ is an $\F_q$-linear set of rank $km$ of $\PG\left(\F_{q^{2km}},\F_{q^{2m}}\right)=\PG(k-1,q^{2m})$ that is maximum $1$-scattered.    
\end{theorem}

\begin{theorem} [see \textnormal{\cite[Theorem 2.3]{bartoli2018maximum}}]\label{th:masxcattkodd2}
Let $\omega \in \F_{q^{2m}} \setminus \F_{q^{m}}$ and $a \in \F_{q^{km}}^*$. Assume $k \geq 3$ odd and define
\[
\U=\{ ax^{q^s}+x\omega \colon x \in  \F_{q^{km}}\} \subseteq \F_{q^{2km}},
\]
for some positive integer $s$.
If 
 \begin{itemize}
        \item $q \equiv 1 \ \mod k$,
       \item $\gcd(s,2m)=\gcd(s,km)=1$,
       \item $\left(N_{q^{km}/q}(a)\right)^{\frac{q-1}{k}} \neq 1$,
    \end{itemize}
then $L_{\U}$ is an $\F_q$-linear set of rank $km$ of $\PG\left(\F_{q^{2km}},\F_{q^{2m}}\right)=\PG(k-1,q^{2m})$ that is maximum $1$-scattered.    
\end{theorem}

\begin{theorem}  [see \textnormal{\cite[Proposition 2.7]{bartoli2018maximum}}] \label{th:maximumscattered3}
    Let $\omega \in \F_{q^{2m}} \setminus \F_{q^{m}}$ and define \begin{equation} \label{eq:deffsab}
f_{s,a,b}:x \in \F_{q^{3m}} \rightarrow ax^{q^s}+bx^{q^{2m+s}} \in \F_{q^{3m}}.
\end{equation}
for some $a,b \in \F_{q^{3m}}^*$ and a positive integer $s$. Let
\[\U=\{ f_{s,a,b}(x)+x\omega \colon x \in  \F_{q^{3m}}\}\subseteq \F_{q^{6m}}.\]

Suppose that $\gcd(s,2m)=1$ and
\[
\frac{f_{s,a,b}(x)}{x} \notin \F_{q^m} \mbox{ for each }x \in \F_{q^{3m}}^*.
\]

Then $L_{\U}$ is an $\F_q$-linear set of rank $3m$ of the projective plane $\PG\left(\F_{q^{6m}},\F_{q^{2m}}\right)=\PG(2,q^{2m})$ that is maximum $1$-scattered.  
\end{theorem}

\begin{theorem} [see \textnormal{\cite[Theorem 2.4]{csajbok2017maximum}}]
    Let $\omega \in \F_{q^{2m}} \setminus \F_{q^{m}}$. There
exist $a, b \in  \F_{q^{3m}}^*$ and an integer $1 \leq s \leq 3m-1$, with $\gcd(s,2m)=1$ such that
\[
\frac{f_{s,a,b}(x)}{x} \notin \F_{q^m} \mbox{ for each }x \in \F_{q^{3m}}^*,
\]
where $f_{s,a,b}$ is defined as in \eqref{eq:deffsab}.
As a consequence, the $\F_q$-linear set $L_{\U}$, defined by
\[
\U=\{ f_{s,a,b}(x)+x\omega \colon x \in  \F_{q^{3m}}\} \subseteq \F_{q^{6m}},
\]
is a maximum $1$-scattered $\F_q$-linear set of the projective plane $\PG(\F_{q^{6m}},\F_{q^{2m}})=\PG(2,q^{2m})$ having rank $3m$.
\end{theorem}

In \cite[Lemma 2.8]{bartoli2018maximum} it was shown that when $q=2$ we may choose $a=1$ in the previous theorem. However  finding explicit values of $a,b$ remains an open problem.


\medskip
\subsection{Maximum \texorpdfstring{$1$}{Lg}-designs}

From a geometric point of view, the construction of $1$-design can be translated in terms of linear sets, by finding a certain number of pairwise disjoint scattered linear sets, as described in the next lemma.

\begin{lemma} \label{lem:designdisjoint}
Let $V$ be a $k$-dimensional $\F_{q^m}$-vector space. Define $\phi_{V}$ be an $\F_{q^m}$-isomorphism between $V$ and $\F_{q^m}^k$ and consider the collection $\U=(\mathcal{U}_1,\ldots,\U_t)$, where $\U_i$ is an $\F_q$-subspace of $V$, such that $\langle \mathcal{U}_1,\ldots,\U_t\rangle_{\F_{q^m}}=V$. Then the following are equivalent:
\begin{enumerate}
    \item $\U$ is a $1$-design in $V$;
    \item $L_{\U_1},\ldots,L_{\U_t}$ are pairwise disjoint scattered $\F_q$-linear sets in $\PG(V,\F_{q^m})$;
    \item $(\phi_V(\U_1),\ldots,\phi_V(U_t))$ is an $1$-design $\Fmnk$ system. 
\end{enumerate} 
\end{lemma}

When the dimension of $V$ is {\bf even}, we can construct maximum $1$-designs having $t$ subspaces, with $1 \leq t \leq q-1$ as following.

\begin{proposition} [see \textnormal{\cite[Proposition 5.16]{santonastaso2022subspace}}] \label{prop:ex1desnorm} 
Let $s$ be a positive integer such that $\gcd(s,m)=1$.
Let $\mu_1,\ldots,\mu_{q-1} \in \F_{q^m}^*$ such that $\N_{q^m/q}(\mu_i) \neq \N_{q^m/q}(\mu_j)$, if $i \neq j$ and define
\[\U_i=\{(x_1,\mu_ix_1^{q^s},\ldots,x_k,\mu_ix_k^{q^s})\colon x_1,\ldots,x_k \in \F_{q^m}\}\subseteq \F_{q^m}^{2k}.\] for $i\in \{1,\ldots,q-1\}$. 
Then $(\mathcal{U}_1,\ldots,\mathcal{U}_{q-1})$ is a maximum $1$-design $[km,2k]_{q^m/q}$ system.
\end{proposition}

The subspaces described before is obtained as the direct sum of special type of linear sets known as linear sets of pseudoregulus type, which were introduced in \cite{marino2007fq} and studied in \cite{lavrauw2013scattered,lunardon2014maximum}.

There exists an other construction of maximum $1$-design for $k$ odd. Indeed, under some assumptions on $k$ and $m$ we can partition $\PG(k-1,q^m)$ into canonical subgeometries. Hence we can construct $1$-designs as follows.

\begin{theorem} [see \textnormal{\cite[Theorem 4.29]{hirschfeld1998projective} \cite[Theorem 5.21]{santonastaso2022subspace}}] 
There exists a partition of $\PG(k-1,q^m)$ into canonical subgeometries if and only if $\gcd(k,m)=1$. In particular, if $m=2$ and $k$ is odd, then there exists a maximum $1$-design of $V=V(k,q^2)$ consisting of $\frac{q^{k}+1}{q+1}$ subspaces.
\end{theorem}

In the following series of results we will solve the problem of the existence of maximum $1$-designs in $V(k,q^m)$, with $1 <t \leq q-1$, for any values of $q,k,m$ such that $km$ is even.

\begin{theorem} \label{th:constructiondisjoint2}
    Let $\omega \in \F_{q^{2m}} \setminus \F_{q^{m}}$ and $a \in \F_{q^{km}}^*$. Let $\mu_1,\ldots,\mu_{q-1} \in \F_{q^m}^*$ such that $\N_{q^m/q}(\mu_i) \neq \N_{q^m/q}(\mu_j)$, if $i \neq j$. Assume $k \geq 3$ odd and define
\[\U_i=\{ a\mu_ix^{q^s}+x\omega \colon x \in  \F_{q^{km}}\} \subseteq \F_{q^{2km}},\]
for some positive integer $s$.
If  \begin{enumerate}[(i)]
       \item $\gcd(m,k)$=1,
       \item $\gcd(s,2m)=1$ and $\gcd(s,km)=k$,
       \item $N_{q^{km}/q^k}(a) \notin \F_q$,
   \end{enumerate} 
then $L_{\U_1},\ldots,L_{\U_{q-1}}$ are pairwise disjoint maximum $1$-scattered $\F_q$-linear set of $\PG(\F_{q^{2km}},\F_{q^{2m}})=\PG(k-1,q^{2m})$ having rank $km$.
\end{theorem}

\begin{proof}
First, we prove that each of $L_{\U_i}$ is a maximum $1$-scattered of $\PG(k-1,q^{2m})$. Clearly, $\dim_{\F_q}(\U_i)=km$. Since $\mu_i \in \F_{q^m}$ and $\gcd(m,k)=1$ we have that $N_{q^{km}/q^k}(\mu_i)=N_{q^{m}/q}(\mu_i) \in \F_q$. Therefore,
\[
N_{q^{km}/q^k}(a\mu_i)=N_{q^{km}/q^k}(a)N_{q^{km}/q^k}(\mu_i)=N_{q^{km}/q^k}(a)N_{q^{m}/q}(\mu_i)
\notin \F_q,\]
by (i). So, for each of $L_{\U_i}$, we are in the hypothesis of Theorem \ref{th:masxcattkodd1} and hence it follows that the $L_{\U_i}$'s are maximum $1$-scattered $\F_q$-linear set of $\PG(k-1,q^{2m})$. \newline
Suppose now by contradiction that $L_{\U_i} \cap L_{\U_j} \neq \emptyset$, for some $i \neq j$. This means that there exist $x,y \in \F_{q^{km}}^*$ and $\lambda \in \F_{q^{2m}}^*$ such that 
\[
(a\mu_ix^{q^s}+x\omega)=\lambda(a\mu_jy^{q^s}+y\omega).
\]
This implies that 
\begin{equation} \label{eq:disjointcondition1}
\left( \frac{a\mu_ix^{q^s}+x\omega}{a\mu_jy^{q^s}+y\omega}\right)^{q^{2m}}= \left( \frac{a\mu_ix^{q^s}+x\omega}{a\mu_jy^{q^s}+y\omega}\right)
\end{equation}
Let 
\[
\omega^2=A+B\omega, 
\]
with $A,B \in \F_{q^m}$, and $A \neq 0$, since $\omega \notin \F_{q^m}$. By \eqref{eq:disjointcondition1}, we obtain that
\[
\mu_i\mu_ja^{q^{2m}+1} x^{q^{s+2m}}y^{q^s}+x^{q^{2m}}yA+(\mu_ia^{q^{2m}} x^{q^{s+2m}}y+\mu_jax^{q^{2m}} y^{q^s}+ x^{q^{2m}}yB ) \omega
\]
\[
=\mu_i\mu_ja^{q^{2m}+1} x^{q^s}y^{q^{s+2m}}+ xy^{q^{2m}}A+(\mu_ja^{q^{2m}} xy^{q^{s+2m}}+\mu_iax^{q^s}y^{q^{2m}}+ xy^{q^{2m}}B)\omega.
\]
Taking into account that $\{1,\omega\}$ is an $\F_{q^m}$-basis of $\F_{q^{2m}}$ and an $\F_{q^{km}}$-basis of $\F_{q^{2km}}$ as well, it follows that 
\begin{equation} \label{eq:disjointcondition2}
    \mu_i\mu_ja^{q^{2m}+1} (x^{q^{2m}}y-xy^{q^{2m}})^{q^s}=(xy^{q^{2m}}-x^{q^{2m}}y)A
\end{equation}
and 
\begin{equation} \label{eq:disjointcondition3}
    a^{q^{2m}}(\mu_ix^{q^{s+2m}}y-\mu_jxy^{q^{s+2m}})+a(\mu_jx^{q^{2m}} y^{q^s}-\mu_ix^{q^s}y^{q^{2m}})=(xy^{q^{2m}}-x^{q^{2m}}y)B.
\end{equation}
Let $u:=x^{q^{2m}}y-xy^{q^{2m}}$. If $u \neq 0$, by \eqref{eq:disjointcondition2}, it is a solution in $\F_{q^{km}}$ of the equation in $z$
\[
z^{q^s-1}=-\frac{A}{\mu_i\mu_ja^{q^{2m}+1}}.
\]

By Condition (ii), and since $(s,km)=k$, the above equation has solutions if and only if 
\[
\N_{q^{km}/q^k} \left( -\frac{A}{\mu_i\mu_ja^{q^{2m}+1}} \right)=1,
\]
i.e.
\[
(-1)^m\N_{q^{km}/q^k} \left( \frac{A}{\mu_i\mu_j} \right)=(\N_{q^{km}/q^k} \left( a \right))^{q^{2m}+1}.
\]
Now, from $A,\mu_i,\mu_j \in \F_{q^m}^*$ and Condition (i), we have $\N_{q^{km}/q^k} \left( \frac{A}{\mu_i\mu_j} \right)=\N_{q^{m}/q} \left( \frac{A}{\mu_i\mu_j} \right)\in \F_q$ and as consequence 
\begin{equation} \label{eq:disjointcondition4}
(\N_{q^{km}/q^k} \left( a \right))^{q^{2m}+1} \in \F_q.
\end{equation}
It follows that $(\N_{q^{km}/q^k} \left( a \right))^{(q^{2m}+1)(q^{2m}-1)}=(\N_{q^{km}/q^k} \left( a \right))^{q^{4m}-1}=1$. 
Note that since $\gcd(k,m)=1$ and $k$ is odd, then $\gcd(k,4m)=1$, and \eqref{eq:disjointcondition4} implies that $\N_{q^{km}/q^k} \left( a \right) \in \F_q$, that contradicts Condition (iii). \newline
Therefore, $u=0$ and so $x=\alpha y$, for some $\alpha \in \F_{q^m}^*$. Substituting in \eqref{eq:disjointcondition3}, we get
\[
(\alpha^{q^s}\mu_i-\alpha \mu_j)(a^{q^{2m}} y^{q^{s+2m}+1}-ay^{q^{s}+q^{2m}})=0.
\]
If $\alpha^{q^s}\mu_i-\alpha \mu_j \neq 0$, then 
\[
a^{q^{2m}} y^{q^{s+2m}+1}-ay^{q^{s}+q^{2m}}=0,
\]
i.e. 
\[
(ay^{q^s-1})^{q^{2m}-1}=1,
\]
and since $\gcd(2m,km)=m$, this implies that $ay^{q^s-1}=\beta$, for some $\beta \in \F_{q^m}^*$. So 
\[
\N_{q^{km}/q^k}(a)=\N_{q^{km}/q^k}(\beta),
\]
and since $\beta \in \F_{q^m}^*$ and $\gcd(k,m)=1$, we get $\N_{q^{km}/q^k}(\beta)=\N_{q^{m}/q}(\beta) \in \F_q$. This implies $\N_{q^{km}/q^k}(a) \in \F_q$, again a contradiction to Condition (iii). \newline
Therefore, $\alpha^{q^s}\mu_i-\alpha \mu_j = 0$, i.e. $\alpha^{q^s-1}=\mu_j/\mu_i$ that implies 
\[
1=\N_{q^{m}/q}(\alpha^{q^s-1})=\N_{q^{m}/q}\left(\frac{\mu_j}{\mu_i}\right),
\]
a contradiction, to our hypothesis on the $\mu_i$'s.
\end{proof}



\begin{theorem} \label{th:constructiondisjoint3}
   Let $\omega \in \F_{q^{2m}} \setminus \F_{q^{m}}$. Let $\mu_1,\ldots,\mu_{q-1} \in \F_{q^m}^*$ such that $\N_{q^m/q}(\mu_i) \neq \N_{q^m/q}(\mu_j)$, if $i \neq j$. Define
\[U_i=\{ a\mu_ix^{q^s}+x\omega \colon x \in  \F_{q^{km}}\} \subseteq \F_{q^{2km}},\]
for some positive integers $s,r$, with $k \geq 3$ odd.
Suppose that  \begin{enumerate}[(i)]
       \item $q \equiv 1 \ \mod k$,
       \item $\gcd(s,2m)=\gcd(s,km)=1$,
       \item $(N_{q^{km}/q}(a))^{\frac{q-1}{k}} \neq 1$.
   \end{enumerate} 
Then $L_{U_1},\ldots,L_{U_{q-1}}$ are pairwise disjoint maximum $1$-scattered $\F_q$-linear sets of $\PG(k-1,q^{2m})$ having rank $km$.
\end{theorem}

\begin{proof}
As in the proof of the previous theorem, we first prove that each of $L_{\U_i}$ is a maximum $1$-scattered of $\PG(k-1,q^{2m})$. Clearly, $\dim_{\F_q}(\U_i)=km$. Since $\mu_i \in \F_{q^m}$, we have that $N_{q^{km}/q}(\mu_i)=(N_{q^{m}/q}(\mu_i))^k$. Therefore,
\[
(N_{q^{km}/q}(a\mu_i))^{\frac{q-1}{k}}=(N_{q^{km}/q}(a))^{\frac{q-1}{k}}((N_{q^{m}/q}(\mu_i))^k)^{\frac{q-1}{k}}=(N_{q^{rm}/q}(a))^{\frac{q-1}{k}} \neq 1,\]
by Condition (iii). So, each $L_{\U_i}$ satisfies the assumptions of Theorem \ref{th:masxcattkodd2} and hence it follows the $L_{\U_i}$'s are maximum $1$-scattered $\F_q$-linear set of $\PG(k-1,q^{2m})$. \newline
Now, arguing as in Theorem \ref{th:constructiondisjoint2}, we have that $L_{\U_i}$ and $L_{\U_j}$ have a common point, for some $i \neq j$ if and only if there exist $x,y \in \F_{q^{km}}^*$ such that equations \eqref{eq:disjointcondition2} and \eqref{eq:disjointcondition3} hold. Choosing again $w \coloneqq xy^{q^{2m}}-yx^{q^{2m}}$, if $w \neq 0$, it turns out to be a solution in $\F_{q^{km}}^*$ of the equation in $z$
\[
z^{q^s-1}=-\frac{A}{\mu_i\mu_ja^{q^{2m}+1}}.
\]
By Condition (ii), the above equation has solutions if and only if 
\[
\N_{q^{km}/q} \left( -\frac{A}{\mu_i\mu_ja^{q^{2m}+1}} \right)=1,
\]
i.e.
\[
(\N_{q^{km}/q} \left( a \right))^{2}=(-1)^m\left(\N_{q^{m}/q} \left( -\frac{A}{\mu_i\mu_j} \right)\right)^k,
\]
implying
\[
\left((\N_{q^{km}/q} \left( a \right))^{\frac{m(q-1)}{k}}\right)^{2}=(-1)^{\frac{m(q-1)}{k}}.
\]
Now, if $\frac{m(q-1)}{k}$ is even, then Condition (iii) implies that $q$ is odd and $\left((\N_{q^{km}/q} \left( a \right))^{\frac{m(q-1)}{k}}\right)^2=-1$. Raising the last equality to the $k$-th power we get a contradiction. While if $\frac{m(q-1)}{k}$ is odd then $q$ is even, and so $\left((\N_{q^{km}/q} \left( a \right))^{\frac{m(q-1)}{k}}\right)=1$, again a contradiction to Condition (iii). It follows that $w=0$ and so $x=\alpha y$, for some $\alpha \in \F_{q^m}^*$. Arguing again as in the proof of Theorem \ref{th:constructiondisjoint2}, we obtain that $ay^{q^s-1}=\beta$, for some $\beta \in \F_{q^m}^*$. So, since $\beta \in \F_{q^m}^*$, it follows that 
\[
(\N_{q^{m}/q}(\beta))^r=\N_{q^{km}/q}(a).
\]
This implies $(\N_{q^{km}/q}(a))^{\frac{m(q-1)}{k}}=1$, a contradiction to Condition (iii). Hence, $\alpha^{q^s}\mu_i-\alpha \mu_j = 0$, i.e. $\alpha^{q^s-1}=\mu_j/\mu_i$ that implies as in the previous proof that 
\[
1=\N_{q^{m}/q}(\alpha^{q^s-1})=\N_{q^{m}/q}\left(\frac{\mu_j}{\mu_i}\right),
\]
a contradiction, to our hypothesis.

\end{proof}
 Finally, in the next result, we prove the existence of maximum $1$-design in $\F_{q^m}^3$, for any $t \leq q-1$, also in the case $m$ even, $m \equiv 0 \, \mbox{mod } 3$ and $q \not \equiv 1 \, \mod 3$. (Note that these parameters are not covered by Theorem \ref{th:constructiondisjoint2} and Theorem \ref{th:constructiondisjoint3}.)

\begin{theorem} \label{th:disjoint4}
    Let $\omega \in \F_{q^{2m}} \setminus \F_{q^{m}}$ and let $\mu_1,\ldots,\mu_{q-1} \in \F_{q^m}^*$ such that $\N_{q^m/q}(\mu_i) \neq \N_{q^m/q}(\mu_j)$, if $i \neq j$. Let $a, b \in  \F_{q^{3m}}^*$ and $1 \leq s \leq 3m-1$, with $\gcd(s,m)=1$ such that
\[
\frac{f_{s,a,b}(x)}{x} \notin \F_{q^m} \mbox{ for each }x \in \F_{q^{3m}}^*,
\]
where $f_{s,a,b}$ is defined as in \eqref{eq:deffsab}.
    Define
\[\U_i=\{ \mu_if_{s,a,b}(x)+x\omega \colon x \in  \F_{q^{3m}}\} \subseteq \F_{q^{6m}}.\]
Then $L_{\U_1},\ldots,L_{\U_{q-1}}$ are pairwise disjoint maximum $1$-scattered $\F_q$-linear sets of the projective plane $\PG(2,q^{2m})$ having rank $3m$.
\end{theorem}

\begin{proof}
First observe that $\dim_{\F_q}(\U_i)=3m$, for each $i \in \{1,\ldots,q-1\}$. Since $\mu_i \in \F_{q^m}$ then also
\[
\frac{\mu_i f_{s,a,b}(x)}{x} \notin \F_{q^m} \mbox{ for each }x \in \F_{q^{3m}}^*
\]
holds. So, each $L_{\U_i}$ satisfies the assumptions of Theorem \ref{th:maximumscattered3} and hence it follows the $L_{\U_i}$'s are maximum $1$-scattered $\F_q$-linear set of $\PG(2,q^{2m})$. \newline 
Suppose now by contradiction that $L_{\U_i} \cap L_{\U_j} \neq \emptyset$, for some $i \neq j$. This means that there exist $x,y \in \F_{q^{3m}}^*$ and $\lambda \in \F_{q^{2m}}^*$ such that 
\[
(\mu_i f_{s,a,b}(x)+x\omega)=\lambda(\mu_j f_{s,a,b}(y)+y\omega).
\]
This implies that 
\begin{equation} \label{eq:disjoint2condition1}
\left( \frac{\mu_if_{s,a,b}(x)+x\omega}{\mu_jf_{s,a,b}(y)+y\omega}\right)^{q^{2m}}= \left( \frac{\mu_if_{s,a,b}(x)+x\omega}{\mu_jf_{s,a,b}(y)+y\omega}\right).
\end{equation}
Let
\[
\omega^2=A+B\omega, 
\]
with $A,B \in \F_{q^m}$, and $A \neq 0$ and taking into account that $\{1,\omega\}$ is an $\F_{q^m}$-basis of $\F_{q^{2m}}$ and an $\F_{q^{3m}}$-basis of $\F_{q^{6m}}$ as well, by \eqref{eq:disjoint2condition1}, we obtain that
\begin{equation} \label{eq:disjoint2condition2}
    \mu_i\mu_j (f_{s,a,b}(x)^{q^{2m}}f_{s,a,b}(y)- f_{s,a,b}(y)^{q^{2m}}f_{s,a,b}(x)) =(xy^{q^{2m}}-x^{q^{2m}}y)A
\end{equation}
and 
\begin{equation} \label{eq:disjoint2condition3}
    \mu_i f_{s,a,b}(x)^{q^{2m}} y+ \mu_j f_{s,a,b}(y) x^{q^{2m}} - \mu_j f_{s,a,b}(y)^{q^{2m}}x-\mu_i f_{s,a,b}(x)y^{q^{2m}} =(xy^{q^{2m}}-x^{q^{2m}}y)B.
\end{equation}
Recalling that $f_{s,a,b}(x)=ax^{q^s}+b{x^{q^{2m+s}}}$, \eqref{eq:disjoint2condition2} reads like
\[
G(w):=\mu_i\mu_j( b^{q^{2m}+1}w^{q^{2m+s}}- b^{q^{2m}}aw^{q^{m+s}} +a^{q^{2m}+1}w^{q^{s}}) +Aw=0,
\]
where $w=xy^{q^{2m}}-x^{q^{2m}}y$. 
Arguing as in the proof of \cite[Proposition 2.7]{bartoli2018maximum} for the case $\mu_i=\mu_j=1$, we get that $w=0$.
This implies that $x=\alpha y$, for some $\alpha \in \F_{q^m}^*$. Putting $x=\alpha y$ and $f_{s,a,b}(x)=ax^{q^s}+b^{x^{q^{2m+s}}}$ in \eqref{eq:disjoint2condition3}, and observing that $f_{s,a,b}(\alpha y)=\alpha^{q^s}f_{s,a,b}(y)$, we have
\[
(\alpha^{q^s}\mu_i-\alpha \mu_j)(f_{s,a,b}(y)^{q^{2m}}y-f_{s,a,b}(y)y^{q^{2m}})=0.
\]
If $f_{s,a,b}(y)^{q^{2m}}y-f_{s,a,b}(y)y^{q^{2m}}=0$, we get that $f_{s,a,b}(y)/y \in \F_{q^{m}}$, a contradiction.
Therefore, $\alpha^{q^s}\mu_i-\alpha \mu_j = 0$, i.e. $\alpha^{q^s-1}=\mu_j/\mu_i$ and so again 
\[
1=N_{q^{m}/q}(\alpha^{q^s-1})=N_{q^{m}/q}\left(\frac{\mu_j}{\mu_i}\right),
\]
a contradiction, to the fact that the $\mu_i$'s have pairwise distinct norm.
\end{proof}

This results allows to state the following existence theorem for maximum 1-design.
\begin{theorem} \label{th:existence1design}
    For any positive integer $k,m \geq 2$, such that $mk$ is even, there exist $L_{\U_1},\ldots,L_{\U_t}$ pairwise disjoint maximum scattered $\F_q$-linear sets in $\PG(V,\F_{q^m})=\PG(k-1,q^m)$, for any $1 \leq t \leq q-1$. In particular, there exists a maximum 1-design $[(mk/2,\ldots,mk/2),k]_{q^m/q}$-system for any $1 \leq t \leq q-1$.
\end{theorem}

\begin{proof}
    The assertion follows by Proposition \ref{prop:ex1desnorm} if $k$ is even. If $k$ is odd the assertion follows by Theorem \ref{th:constructiondisjoint2} if $\gcd(m,k)$=1 and by Theorem \ref{th:constructiondisjoint3} if $q \equiv 1 \ \mod k$.
    If $k,m,q$ do not satisfy one of the above mentioned assumptions, we get the existence of $L_{\U_1},\ldots,L_{\U_t}$ pairwise disjoint maximum scattered $\F_q$-linear sets in $\PG(V,\F_{q^m})=\PG(k-1,q^m)$, for any $1 \leq t \leq q-1$, by decomposing $V$ as the direct sum of at least one vector subspace of dimension 3 and a suitable number of vector subspaces of dimension 2 and 3, using Theorem \ref{th:disjoint4} and then applying Proposition \ref{cor:directsum}. The last part follows by Lemma \ref{lem:designdisjoint}.
\end{proof}

\subsection{Maximum \texorpdfstring{$1$}{Lg}-design with $t>q-1$} \label{subsec:maximum1designt>q-1}
The question remains as to whether the values of $t$ covered in Theorem \ref{th:existence1design} are optimal. An upper bound on $t$ is the following.
\begin{proposition} [see \textnormal{\cite[Corollary 5.22]{santonastaso2022subspace}}]
Let $(U_1,\ldots,U_t)$ be a maximum $1$-design in $V=V(k,q^m)$.
Consider the associated $\fq$-linear sets $L_{U_1},\ldots,L_{U_t}$. Then 
    $t \leq (q-1)\left(\frac{q^{\frac{mk}{2}}+1}{q^m-1}\right)$,
and the equality holds if and only if $m=2$, $k$ is odd, $L_{U_i}$ is a Baer subgeometry of $\PG(k-1,q^2)$ for any $i$ and $L_{U_1},\ldots,L_{U_t}$ is a partition of $\PG(k-1,q^2)$.
\end{proposition}

We will show now some contructions of maximum $1$-design with $t>q-1$ in the case $m=4$ and $k \equiv 0 \ \mod 3$. In this case we show the existence of a set of $t=2(q-1)$ pairwise disjoint maximum scattered linear sets of rank $6$ in $\PG(k-1,q^4)$, with $k \equiv 0 \ \mod 3$. This will require some results and techniques from the study of algebraic objects known as finite (pre)semifields, and so we recall the following notions.

Finite \textbf{presemifields} are division algebras with a finite number of elements in which multiplication is not assumed to be associative. If a presemifield contains a multiplicative identity, then
it is called a \textbf{semifield}.
It is possible to identify the elements of a $n$-dimensional semifield over $\F_q$ with the elements of the
field extension $\fqn$.  We refer to \cite{lavrauw2011finite}, for details on semifields; in particular the operations of isotopy and Knuth-equivalence, which preserve the property of having no non-trivial zero divisors. Two multiplications $S$ and $S'$ are \textbf{isotopic} if there exist invertible additive maps $f,g,h$ such that $S'(x,y)= f(S(g(x),h(y))$. The Knuth operations arise from the following action of the symmetric group $S_3$. Given an element $\pi\in S_3$ and a multiplication $S(x,y)$, the multiplication $S^\pi(x,y)$ is defined as one satisfying 
\[
\mathrm{Tr}_{q^6/q}(S^{\pi}(x_1,x_2)x_3) = \mathrm{Tr}_{q^6/q}(S(x_{\pi(1)},x_{\pi(2)})x_{\pi(3)}).
\]
Two multiplications which are equivalent under a combination of these actions are said to be \textbf{ Knuth equivalent}.

}

We start with a useful lemma. Let us denote 
\[U_a=\{ ax^{q^3}+x\omega \colon x \in  \F_{q^{6}}\} \subseteq \F_{q^{12}}.\]

\begin{lemma} \label{lem:existence1designdouble}
The linear sets $L_{U_a}$ and $L_{U_b}$ are disjoint if and only if there are no solutions to the system $S(x,y)=T(x,y)=0$, $xy\ne 0$, where
\begin{align*} S(x,y)&:= (xy^{q^{4}}-x^{q^{4}}y)A-(a^{q^4}bx^{q}y^{q^3}-ab^{q^4}x^{q^3}y^{q})\\
T(x,y)&:= 
(xy^{q^{4}}-x^{q^{4}}y)B- (a^{q^{4}}x^{q}y-ax^{q^3}y^{q^{4}})-(bx^{q^{4}} y^{q^3}-b^{q^4}xy^{q}).
\end{align*}
\end{lemma}

\begin{proof}
The proof is completely analogous to the proof of Theorem \ref{th:constructiondisjoint2} up to equations (\ref{eq:disjointcondition2}) and (\ref{eq:disjointcondition3}).
\end{proof}

\begin{theorem} 
There exists a set of $2(q-1)$ pairwise disjoint maximum scattered $\F_q$-linear sets in $\PG(2,q^4)$ for any $q$ odd with $q\equiv ~1~ \mod 3$.
\end{theorem}

\begin{proof}
Let $a\in \F_{q^6}$ be such that $N_{q^6/q^3}(a)\notin \F_q$, and $b$ such that $a^{q^3}b=\mu\in \F_{q^2}$. We claim that $L_{U_a}$ and $L_{U_b}$ are disjoint. Let $S(x,y)$ be as in Lemma \ref{lem:existence1designdouble}. First we suppose that $\mu^{q+1}\ne A^{q+1}$, and show that there are no non-trivial solutions to $S(x,y)=0$, and hence $L_{U_a}$ and $L_{U_b}$ are disjoint.
    
    Note that $S(x,y)$ is an $\fq$-bilinear map, and so we are claiming that $S(x,y)$ defines the multiplication of a (pre)semifield. 
A straightforward calculation verifies that 
    \[
    S^{(23)}(x,y)=   A(x^{q^2}y^{q^2}-x^{q^4}y)-\mu(a^{q-1}x^{q^4}y^{q^3}-a^{q^5-1}x^{q^2}y^{q^5}),
    \] 
    
    Thus $S(x,y)$  is Knuth-equivalent to 
    \[
    S_1(x,y):= S^{(23)}(x^{q^4},y) = (Ay+\mu a^{q^3-q^4}y^{q^3})^{q^2}x-(Ay+\mu a^{q-1}y^{q^3})x^{q^2},
    \] 
    
    which defines a presemifield if and only if 
\[
    (Ay+\mu a^{q^3-q^4}y^{q^3})^{q^4+q^2+1}-(Ay+\mu a^{q-1}y^{q^3})^{q^4+q^2+1}\ne 0\quad\forall~y\ne 0.
    \]
    This condition holds if and only if the multiplication
        \begin{align*}
           S_2(x,y)&:= (Ay+\mu a^{q^3-q^4}y^{q^3})x-(Ay+\mu a^{q-1}y^{q^3})x^{q^2} \\
           &= A(x-x^{q^2})y+\mu(a^{q^3-q^4}x-a^{q-1}x^{q^2})y^{q^3}
        \end{align*}

defines a presemifield. 
This multiplication is in turn isotopic to 
\begin{align*}
S_3(x,y)&:=\left(\frac{\mu^{-1}a^{q^4-q^3}}{1-a^{(q^3+1)(q-1)}}\right)S_2(y,x-\mu A^{-1}a^{q-1}x^{q^3})\\&= (\alpha y+\beta y^{q^2})x+yx^{q^3},
\end{align*}
where 
\[
\alpha = \frac{a^{q^4-q^3}(A^{q+1}-\mu^{q+1})}{\mu A^{q}(1-a^{(q^3+1)(q-1)})}
 ;\quad \beta = \frac{ a^{q^4-q^3}(\mu^{q+1}a^{(q^3+1)(q-1)}-A^{q+1})}{\mu A^{q}(1-a^{(q^3+1)(q-1)})}.
\]

This is an algebra of the form $\mathbb{S}(\alpha,0,\beta,0,q^2)$ defined in \cite{marino2011towards} and studied further in \cite{lavrauw2011fq}. There it was shown that there exist $\lambda$ such that $\mathbb{S}(\lambda,0,\lambda,0,q^2)$ is a presemifield if $q$ is odd and $q\equiv~ 1~ \mod 3$ \cite[Theorem 5.3]{lavrauw2011fq}; namely, $\lambda$ satisfying $\lambda^{q^3+1}= \frac{1}{2(1-\epsilon)}$ for $\epsilon$ a primitive sixth root of unity. 
Note that our assumption that $\mu^{q+1}\ne A^{q+1}$ implies that $\alpha\ne 0$.
Note that, if 
\begin{equation} \label{eq:condsemifield1}
\alpha^{q^4+q^2+1}=\beta^{q^4+q^2+1}
\end{equation}
 and 
 \begin{equation} \label{eq:condsemifield2}
     \alpha^{q^3+1}=\lambda^{q^3+1},
 \end{equation}
 then $\mathbb{S}(\lambda,0,\lambda,0,q^2)$ is isotopic to $\mathbb{S}(\alpha,0, \beta,0,q^2)$, as we will now demonstrate.
Conditions \eqref{eq:condsemifield1} and \eqref{eq:condsemifield2} hold if and only if there exist nonzero $\eta_1,\eta_2\in \F_{q^6}$ such that $\alpha=\lambda \eta_1^{q^3-1}=\beta \eta_2^{q^2-1}$, in which case $(\eta_1^{q^3}\eta_2)^{-1} S_3(\eta_1 x,\eta_2 y)$ is the multiplication of $\mathbb{S}(\lambda,0,\lambda,0,q^2)$. Hence it suffices to show that there exist choices for $a$ and $A$ such that these conditions are satisfied. Note that we are indeed free to choose $A$ as any nonzero element of $\F_{q^2}$, since it is the norm of $\omega$ an arbitrary element of $\F_{q^4}\backslash\F_{q^2}$.

We choose $a$ such that $\rho := a^{(q^3+1)(q-1)}\in\Fq\backslash\{0,1\}$. Note that this necessitates that $\rho$ is a primitive third root of unity. We set $\tau= \frac{A^{q+1}}{\mu^{q+1}}$. Then condition $\alpha^{q^4+q^2+1}=\beta^{q^4+q^2+1}$ is satisfied if and only if $(\tau-1)^3=(\rho-\tau)^3$.
 The condition $\alpha^{q^3+1}=\lambda^{q^3+1}$ becomes
\[
\frac{\rho}{\tau}\left(\frac{\tau-1}{\rho-1}\right)^2 = \frac{1}{2(1-\epsilon)}.
\]
The choices $\rho=\epsilon^2$ and $\tau = \epsilon/2$ are sufficient to prove this existence (note that the primitive sixth roots of unity are precisely the roots of the polynomial $x^2-x+1$).

Now we suppose $\mu^{q+1}=A^{q+1}$. Then viewing $S(x,y)$ as a linearized polynomial in $x$, we get that
\[
S(x,y) = \left(x+\left(\frac{\mu}{A^q a^{q-1}}\right)x^{q^3}\right) \circ (Ay^{q^4}x-\mu a^{q^4-q^3} y^{q^3}x^q),
\]
where $\circ$ denotes composition as polynomials in $x$. The first factor of this composition has no nontrivial solutions in $\F_{q^6}$, since $\left(\frac{\mu}{A^q a^{q-1}}\right)^{q^3+1}\ne 1$. Hence the nontrivial solutions to $S(x,y)=0$ are when $x=\frac{\lambda y^{q^3}}{\tau a^{q^3}}$, where $\tau\in \F_{q^2}$ is such that $\mu=A\tau^{q-1}$, and $\lambda\in \Fq$. 

Plugging this value of $x$ into $T(x,y)=0$, we get that $a^{q^3-q}=y^{(q+1)(q^3-1)}$, contradicting the fact that that $a^{(q-1)(q^3+1)}\ne 1$. Hence $L_{U_a}$ and $L_{U_b}$ are disjoint.



Finally we take the set $\{L_{\mu_i a}:i=1,\ldots q-1\}\cup \{L_{\mu_i b}:i=1,\ldots q-1\}$,
with the $\mu_i$ as in Theorem \ref{th:constructiondisjoint2}. We have shown that these linear sets are indeed pairwise disjoint, and are maximum scattered by Theorem \ref{th:masxcattkodd1}, proving the claim.
\end{proof}

Under the assumptions of the above theorem, by using Proposition \ref{cor:directsum}, we get the existence of $2(q-1)$ pairwise disjoint maximum scattered $\F_q$-linear sets in $\PG(3r-1,q^4)$.
\begin{corollary} \label{cor:maximum1designt>q-1particular}
     There exists a set of $2(q-1)$ pairwise disjoint maximum scattered $\F_q$-linear sets in $\PG(3r-1,q^4)$ for any $r\geq 1$, $q$ odd with $q\equiv ~1~ \mod 3$.
\end{corollary}

For $q \in \{2,3\}$, we have the following constructions, with the aid of MAGMA.

\begin{theorem} 
\label{thm:sporadic1des}
    There exist two disjoint maximum scattered $\F_q$-linear sets in $\PG(2,2^4)$, and a set of six pairwise disjoint maximum scattered $\F_q$-linear sets in $\PG(2,3^4)$. In particular, there exists a maximum 1-design $[(6,6),3]_{2^4/2}$-system, and a maximum 1-design $[(6,6,6,6,6,6),4]_{3^4/3}$-system.
\end{theorem}

\begin{proof}


For $q=2$, let $a$ be any element of $\F_{2^6}$ such that $N_{2^6/2^3}(a)\notin \F_2$, and $b$ such that $(ab)^3a^{q^4+q^2+1}=1$. Then it can be verified by using MAGMA applied to Lemma \ref{lem:existence1designdouble} that $L_{U_a}$ and $L_{U_b}$ are disjoint.

For $q=3$, let $a_i$ be any element of $\F_{3^6}$ such that $a_i^{q+1}$ is a root of the polynomial $(z^3-x+1)(z^3+z^2-1)$. Again it can be verified that $L_{U_{a_i}}$ and $L_{U_{a_j}}$ are disjoint maximum scattered linear sets for all $i\ne j$.
\end{proof}

Further computations show that there exist at least nine pairwise disjoint maximum scattered $\F_q$-linear sets in $\PG(2,4^4)$, and twelve pairwise disjoint maximum scattered $\F_q$-linear sets in $\PG(2,5^4)$; in each case (except $q=2$) we reach $t=3(q-1)$. These searches restricted to linear sets of the form $L_{U_a}$ as in the previous theroem. However, we do not have a general construction at present. The general problem of existence and maximum length of maximum $1$-designs remains open.

\begin{open} Do there exist $t$ pairwise disjoint maximum scattered $\F_q$-linear sets in $\PG(k-1,q^m)$, for  $t>q-1$ in general? What is the maximum  value  of $t$ for which such linear sets exist?
\end{open}


\section{Two-intersecting sets, strongly regular graphs and two-weight linear codes from maximum 1-designs}

Maximum $1$-designs give examples of two-weight intersecting sets as proved in the following.

\begin{proposition} [see \textnormal{\cite[Corollary 5.14]{santonastaso2022subspace}}]
For a maximum $1$-design $(\U_1,\ldots,\U_t)$ in $V=V(k,q^m)$, consider the associated $\fq$-linear sets $L_{\U_1},\ldots,L_{\U_t}$.
Then
\[ \left| H \cap (L_{U_1}\cup \ldots\cup L_{U_t}) \right| \in \left\{t \left(\frac{q^{\frac{m(k-2)}{2}}-1}{q-1}\right),t \left(\frac{q^{\frac{m(k-2)}{2}}-1}{q-1}\right)+q^{\frac{m(k-2)}{2}}\right\}, \]
 for any hyperplane $H$ of $\PG(k-1,q^m)=\PG(V,\F_{q^m})$.
In particular, 
\begin{itemize}
\item $L_{U_1}\cup \ldots\cup L_{U_t}$ is a two-intersection set with respect to hyperplanes if and only if $t< (q-1)\frac{q^{\frac{mk}2}+1}{q^m-1}$;
\item if $t= (q-1)\frac{q^{\frac{mk}2}+1}{q^m-1}$ then $L_{U_1}\cup \ldots\cup L_{U_t}=\PG(k-1,q^m)$.
\end{itemize}
\end{proposition}

As a consequence of Theorem \ref{th:existence1design}, we get the following constructions of two-intersection sets.

\begin{corollary}
    Let $t$ be a positive integer such that $1 \leq t \leq q-1$. For any positive integers $k,m \geq 2$ such that $mk$ is even, there exists a set $S \subseteq \PG(k-1,q^m)$ such that
    \[ \left| H \cap S \right| \in \left\{t \left(\frac{q^{\frac{m(k-2)}{2}}-1}{q-1}\right),t \left(\frac{q^{\frac{m(k-2)}{2}}-1}{q-1}\right)+q^{\frac{m(k-2)}{2}}\right\}. \]
\end{corollary}

It is also well-known the connection between two-intersection sets with respect to the hyperplanes and strongly regular graphs, see e.g.\ \cite{calderbank1986geometry}.

Let $\mathcal{G}$ be a regular graph with $v$ vertices and degree $K$. $\mathcal{G}$ is said to be \textbf{strongly regular} if there exist two integers $\lambda$ and $\mu$ such that every two adjacent vertices have $\lambda$ common neighbours and every two non-adjacent vertices have $\mu$ common neighbours. Also, we say that $(v,K,\lambda,\mu)$ are its parameters. Starting from a two-intersection set with respect to the hyperplanes, it is possible give constructions of strongly regular graphs a following.

Let $S$ be a set of points in $\PG(k-1,q)$ such that $\langle \mathcal{P} \rangle=\PG(k-1,q)$.
Embed $\PG(k-1,q)$ as a hyperplane $H$ of $\PG(k,q)$.
Let $\Gamma(S)$ be the graph whose vertices are the points in $\PG(k,q)\setminus \mathcal{H}$ and two vertices $P$ and $Q$ are adjacent if the line $PQ$ meets $H$ in a point of $S$.
In the case that $S$ is a two-intersection set with respect to the hyperplanes, then $\Gamma(S)$ is a strongly regular graph. 
So applying the above described construction to the two-intersecting set associated to a maximum $1$-design we obtain a strongly regular graph with the following parameters.

\begin{proposition}
[see \textnormal{\cite[Corollary 5.27]{santonastaso2022subspace}}]   \label{cor:srg}
Let $t$ be a positive integer such that $1\leq t <(q-1)\frac{q^ {\frac{mk}{2}}+1}{q^m-1}$. Let $m,k \geq 2$ positive integers such that $mk$ is even. Let $(\U_1,\ldots,\U_t)$ be a maximum $1$-design in $V=V(k,q^m)$.
Define $$
M=t\left(\frac{q^{\frac{km}2}-1}{q-1}\right);\quad w_0=t \left(\frac{q^{\frac{m(k-2)}{2}}-1}{q-1}\right); \quad w_1=t \left(\frac{q^{\frac{m(k-2)}{2}}-1}{q-1}\right)+q^{\frac{m(k-2)}{2}}.
$$
Then $\Gamma(L_{\U_1}\cup\ldots\cup L_{\U_t})$ is a strongly regular graph with parameters $(v,K,\lambda,\mu)$, where
\begin{itemize}
    \item $v=q^{mk}$;
    \item $K=M(q^m-1)$;
    \item $\lambda=K^2+3K-q^m(1+K)(2M-w_1-w_0)+q^{2m}(M-w_1)(M-w_0)$;
    \item $\mu=\frac{q^{2m}(M-w_1)(M-w_0)}{q^{mk}}$.
\end{itemize}
\end{proposition}

Taking into account Theorem \ref{th:existence1design}, we get the following constructions of strongly regular graphs.

\begin{corollary}
    For any $t$ positive integer such that $1\leq t <q-1$ and any $m,k \geq 2$ positive integers such that $mk$ is even, there exist a strongly regular graph with parameters $(v,K,\lambda,\mu)$, where $v,K,\lambda,\mu$ are as in Proposition \ref{cor:srg}
\end{corollary}

Finally, we can consider the linear codes associated with the linear sets defined by a maximum $1$-design. It is well-known that set of points (considered with a multiplicity function) in the projective space corresponds to linear Hamming metric codes, see \cite{vladut2007algebraic}. For a set of points $S$ we can construct a matrix $G$ by taking as columns a representative of each point in $S$. In this way, it is possible to define an hamming metric code $\C$ having $G$ as a generator matrix. Considering the set of points the one defined by the linear sets associated to a maximum $1$-design, we get linear codes with the following parameters.

\begin{proposition} [see \textnormal{\cite[Corollary 5.24]{santonastaso2022subspace}}] \label{prop:parametertwoweightcodes}
Let $(U_1,\ldots,U_t)$ be a maximum $1$-design in $V=\F_{q^m}^k$, with $t< (q-1)\frac{q^{\frac{mk}2}+1}{q^m-1}$.
Consider the associated $\fq$-linear sets $L_{U_1},\ldots,L_{U_t}$ and let $S=L_{U_1}\cup \ldots\cup L_{U_t}$.
Then a linear code $\C \subseteq \F_{q^m}^M$ associated to $S$ has dimension $k$, minimum distance $d=M-w_1$, and weight enumerator polynomial
\[ 1+(q^m-1)h_1z^{M-w_1}+(q^m-1)h_0z^{M-w_0}, \]
where
\begin{itemize}
    \item $w_0=t \frac{q^{\frac{m(k-2)}{2}}-1}{q-1}$;
    \item $w_1=(t-1) \frac{q^{\frac{m(k-2)}{2}}-1}{q-1}+\frac{q^{\frac{m(k-2)}{2}+1}-1}{q-1}$;
    \item $M=t\frac{q^{\frac{km}2}-1}{q-1}$;
    \item $h_1=t\frac{(q^{\frac{mk}2}-1)(q^{m(k-1)}-1)-(q^{\frac{m(k-2)}2}-1)(q^{mk}-1)}{(q^m-1)(q-1)q^{\frac{m(k-2)}2}}$;
    \item $h_0=\frac{q^{mk}-1}{q^m-1}-h_1$;
    \item $d=M-w_1$.
\end{itemize}
\end{proposition}

These codes are two-weight codes and we can completely determine their weight enumerator. Theorem \ref{th:existence1design} ensures the following existence result.

\begin{corollary}
 For any $t$ positive integer such that $1\leq t <q-1$ and any $m,k \geq 2$ positive integers such that $mk$ is even, there exist a linear code $\C \subseteq \F_{q^m}^M$ with dimension $k$, minimum distance $d=M-w_1$, and weight enumerator polynomial
\[ 1+(q^m-1)h_1z^{M-w_1}+(q^m-1)h_0z^{M-w_0}, \]
where $w_0,w_1,M,h_1,h_0,d$ are as in Proposition \ref{prop:parametertwoweightcodes}.
\end{corollary}

\begin{remark}
By the constructions in Subsection \ref{subsec:maximum1designt>q-1}, we also get two-intersection sets, strongly regular graphs and two-weight linear codes with parameters described in the above corollaries for any $t \leq 2(q-1)$, $m=4$ and $k \equiv 0 \ \mod 3$, as well as some sporadic larger values of $t$ for some small $q$. For example taking $q=2$, $k=3$, $m=4$, $t=2$, we obtain an $\mathrm{srg}(4096,1890,874,870)$.
\end{remark}

\section{Constructions and existence of MSRD code and \texorpdfstring{$h$}{Lg}-design}

In this section, we will explore the consequences of the previous sections for the existence of $h$-design and MSRD codes.

Considering a system associated to a Linerized Reed-Solomon code and Twisted Linearized Reed-Solomon code as already proved in \cite{santonastaso2022subspace,neri2023geometry}, we get the following construction of $(k-1)$-designs.

\begin{theorem} [see \textnormal{\cite[Theorem 5.2]{santonastaso2022subspace} \cite[Proposition 4.3]{neri2023geometry}}] 
\label{th:example2(n-1)des}
     Let $\mu_1,\ldots,\mu_{t} \in \F_{q^m}^*$ and suppose that $\N_{q^m/q}(\mu_i) \neq \N_{q^m/q}(\mu_j)$ if $i \neq j$. Let $\eta \in \F_{q^m}$ with the property that $\N_{q^m/q}(\eta)(-1)^{km} \notin G$, with $G$ the multiplicative subgroup of $\F_{q}^*$ generated by $\{\N_{q^m/q}(\mu_1),\ldots,\N_{q^m/q}(\mu_t)\}$. Let $S$ be an $n$-dimensional $\F_q$-subspace of $\F_{q^m}$. If $k\leq tn$ then 
    the ordered set of the $\fq$-subspaces $(\U_1,\ldots,\U_{t})$, where
    \[
    \U_i=\{ (x+\eta \N_{\sigma}^k(\mu_i)\sigma^k(x),\sigma(x)\N_{\sigma}^1(\mu_i),\sigma^2(x)N_{\sigma}^2(\mu_i),\ldots, \sigma^{k-1}(x)\N_{\sigma}^{k-1}(\mu_i)) \colon x \in  S\} \subseteq \F_{q^m}^k
    \]
    is $(k-1)$-design $[\bfn,k]_{q^m/q}$-system, with $\bfn=(n,\ldots,n)$. If $n=m$, then it is a maximum $(k-1)$-design.
\end{theorem}
\medskip

\medskip 
\begin{remark}
Note that condition $(-1)^{km}\N_{q^m/q}(\eta) \notin G$ in Theorem \ref{th:example2(n-1)des} leads to a restriction on $t$. As a consequence, we can construct subspace designs with $\eta \neq 0$ if $t$ is not larger of the size of biggest proper subgroup of $\F_{q}^*$, that is $(q-1)/r'$, where $r'$ is the smallest prime diving $q-1$.
If $G=\F_{q}^*$, the only possibly choice for $\eta$ is $\eta=0$ and in this case we can go until $t=q-1$, see also \cite[Proposition 6.8]{neri2022twisted}.
\end{remark}

It is possible construct a $(k-1)$-design having a number of blocks greater than $q-1$ in the following way.
\begin{proposition} \label{prop:optimaldesignlargerblocks}
    Let $\mu_1,\ldots,\mu_{t} \in \F_{q^m}^*$ and suppose that $\N_{q^m/q}(\mu_i) \neq \N_{q^m/q}(\mu_j)$ if $i \neq j$. Let $S$ be an $n$-dimensional $\F_q$-subspace of $\F_{q^m}$ and $S_{j_i,i} \neq \{0\}$ be $\F_q$-subspaces of $S$ such that 
    \[
    S=S_{1,i} \oplus \cdots \oplus S_{j_i,i}, 
    \]
    for each $i$.
    If $k\leq tn$ then 
    the ordered set of the $\fq$-subspaces $(\U_{1,1},\ldots,\U_{j_1,1},\ldots,\U_{1,t},\ldots,\U_{j_t,t})$, where
    \[
    \U_{l,i}=\{ (x,\sigma(x)\N_{\sigma}^1(\mu_i),\sigma^2(x)N_{\sigma}^2(\mu_i),\ldots, \sigma^{k-1}(x)\N_{\sigma}^{k-1}(\mu_i)) \colon x \in  S_{l,i}\} \subseteq \F_{q^m}^k
    \]
    is $(k-1)$-design in $\F_{q^m}^k$
\end{proposition}

\begin{proof}
    Consider 
    \[
    \U_i=\{ (x,\sigma(x)\N_{\sigma}^1(\mu_i),\sigma^2(x)N_{\sigma}^2(\mu_i),\ldots, \sigma^{k-1}(x)\N_{\sigma}^{k-1}(\mu_i)) \colon x \in  S\} \subseteq \F_{q^m}^k.
    \]
    By Theorem \ref{th:example2(n-1)des}, we know that $(\U_1,\ldots,\U_t)$ is a $(k-1)$-design in $\F_{q^m}^k$. Since 
    \[
    \dim_{\F_q}(\U_{1,i} \cap H)+\ldots +\dim_{\F_q}(\U_{j_i,i} \cap H) \leq \dim_{\F_q}(\U_{i} \cap H),
    \]
    for each $i$ and any hyperplane $H$ of $\F_{q^m}^k$, it follows that also the $U_{l,i}'s$ form a $(k-1)$-design in $\F_{q^m}^k$. Finally, since 
    \[
  \langle \U_{1,1},\ldots,\U_{j_t,t} \rangle_{\F_{q^m}}=\langle \U_{1},\ldots,\U_{t} \rangle_{\F_{q^m}}=\F_{q^m}^k,
    \]
    we have the assertion.
\end{proof}

Since when $m \geq n_i$, optimal subspace designs correspond to $(k-1)$-design (cf. Corollary \ref{cor:optimalk-1design}) using Proposition \ref{prop:optimaldesignlargerblocks}, we get the following result.

\begin{corollary}
    There exists an MSRD $\Fmnk$ code with $m \geq n_i$, for any $k<\sum_{i=1}^t n_i \leq (q-1)m$.
\end{corollary}

Recalling that the Delsarte dual system $\U^{\perp}$ of a system $\U$ is the system associated to a code $\C \in \Phi([\U])$. So from Theorem \ref{Th:sdesignMSRD} we immediately get the following geometric property on $\U^{\perp}$.

\begin{corollary}
Let $\mathcal{U}$ be an $\Fmnkd$ system with $d \geq 2$. Then $\mathcal{U}$ is an $(h,r)$-subspace design if an only if an its Delsarte dual $[\bfn,N-k]_{q^m/q}$ system $\mathcal{U}^{\perp}$ is an $(N-k-r+h-1,N-r-2)$-design.
   
\end{corollary}

Another consequence of Theorem \ref{Th:sdesignMSRD}, is that a Delsarte dual of a system that is a maximum $h$-design is also an optimal $h'$-design as well.

\begin{corollary} \label{cor:dualofmaximumdesign}
    Let $(\mathcal{U}_1,\ldots,\mathcal{U}_t)$ be a maximum  $h$-design $[\bfn,k]_{q^m/q}$ system, with $n=\frac{mk}{h+1}$. Then any Delsarte dual $\mathcal{U}^{\perp}$ of $\mathcal{U}$ is an optimal $(tm-h-2)$-design $[\bfn,N-k]_{q^m/q}$ system.
\end{corollary}

\begin{proof}
    Since $(\mathcal{U}_1,\ldots,\mathcal{U}_t)$ is a maximum  $h$-design $[\bfn,k]_{q^m/q}$ system, by Theorem  \ref{Th:sdesignMSRD} any code $\C \in \Psi([\mathcal{U}])$ is an MSRD $[\bfn,k]_{q^m/q}$ code and so any $\C^{\perp}$ is an MSRD $[\bfn,N-k]_{q^m/q}$ code as well. Therefore, by applying Theorem \ref{Th:sdesignMSRD} to $\C^{\perp}$, we have that any Delsarte dual $\mathcal{U}^{\perp}$ of $\mathcal{U}$ is an optimal $h'$-design $[\bfn,N-k]_{q^m/q}$ system, with $h'$ such that $\frac{mk}{h+1}=\frac{m(N-k)}{h'+1}$ and hence the assertion follows.
\end{proof}

\begin{theorem} [see \textnormal{\cite[Theorem 5.8]{santonastaso2022subspace}}] \label{th:applincollamento}
    If $h+1$ divides $k$, $m \geq h+1$ and $t\leq q-1$, then there exists a maximum $h$-design $\left[\left(\frac{mk}{h+1},\ldots,\frac{mk}{h+1}\right),k\right]_{q^m/q}$ system.
\end{theorem}


This existence result was the only known construction for $1<h<k-1$. Hence the following theorem provides new examples.

\begin{theorem}
    For $m \geq 4$ even, $k \geq 3$ odd and $t \leq q-1$, there exist a optimal $(tm-3)$-design $\left[(\frac{km}{2},\ldots,\frac{km}{2}),\frac{k(tm-2)}{2}\right]_{q^m/q}$ system, which cannot be obtained by Theorem \ref{th:applincollamento}.
\end{theorem}

\begin{proof}
    For $m \geq 4$ even and $k \geq 3$ odd, by Theorem \ref{th:existence1design}, there exists $\U$ a maximum $1$-design $\Fmnk$ system. Let consider a Delsarte dual system $\mathcal{U}^{\perp}$ of $\U$. By Corollary \ref{cor:dualofmaximumdesign}, for $h=1$, we get that  $\mathcal{U}^{\perp}$ is a maximum $(tm-3)$-design. Moreover, since $tm-2$ does not divide $\frac{k(tm-2)}{2}$ then this system cannot be obtained by by Theorem \ref{th:applincollamento}.
\end{proof}

From Theorem \ref{th:existence1design} and Theorem \ref{Th:sdesignMSRD}, we get immediately the following existence result. 
\begin{theorem} 
    For any positive integer $k,m \geq 2$, such that $mk$ is even and any $1 \leq t \leq q-1$, there exist a MSRD $[(mk/2,\ldots,mk/2),k]_{q^m/q}$ code.
\end{theorem}

Moreover, from Corollary \ref{cor:maximum1designt>q-1particular} and Theorem \ref{Th:sdesignMSRD}, we get MSRD codes having a number of blocks greater than $q-1$.

\begin{theorem}
For any positive integer $r$ and $q$ odd with $q\equiv ~1~ \mod 3$, there exists an MSRD $[(6r,\ldots,6r),3r]_{q^4/q}$ code.
\end{theorem}

It remains an open problem to construct MSRD codes for the remaining values of $m,k$. However we note that even in the rank metric this question remains open, and the question in the sum-rank metric appears to be strictly more difficult.
\begin{open}
Construct MSRD codes with $n=mk/(h+1)$, for the values of $m,k,h$ not covered by the previous results.
\end{open} 

\medskip

\section*{Acknowledgements}
The authors are very grateful to Olga Polverino and Ferdinando Zullo for carefully reading the paper and giving several suggestions. John Sheekey would like to thank the University
of Campania ``Luigi Vanvitelli” and his coauthor for his
hospitality during the preparation of this article.
Paolo Santonastaso is very grateful for the hospitality of the School of Mathematics and Statistics, University College Dublin, Dublin, Ireland, where he was a visiting PhD student for 1 month during the preparation of this paper.
This work was supported by the Italian National Group for Algebraic and Geometric Structures and their Applications (GNSAGA - INdAM).

\bibliographystyle{abbrv}
\bibliography{biblio}
\end{document}